\documentclass[11pt]{article}
\pdfoutput=1
\usepackage{amssymb}
\usepackage{graphicx}
\usepackage{url}
\usepackage{setspace}
\usepackage[pdftex,bookmarksnumbered,bookmarksopen,
colorlinks,citecolor=blue,linkcolor=blue]{hyperref}
\usepackage{tablefootnote}
\usepackage{framed}
\usepackage{xcolor}
\usepackage{soul}
\usepackage{longtable}

\usepackage{times}
\usepackage{enumitem}
\usepackage{varwidth}
\usepackage{graphicx}
\usepackage{wrapfig}
\usepackage{enumerate}

\usepackage[utf8]{inputenc} 
\usepackage[T1]{fontenc}    
\usepackage{url}            
\usepackage{booktabs}       
\usepackage{amsfonts}       
\usepackage{nicefrac}       
\usepackage{microtype}      

\usepackage[tight]{subfigure}
\usepackage{graphicx}
\usepackage{appendix}
\usepackage{amsmath,amsfonts,amsthm}
\usepackage{algorithm,algorithmic}
\usepackage{mdwlist}
\usepackage{xspace}
\usepackage{color}
\usepackage{mathrsfs}

\usepackage{booktabs}
\usepackage{comment}
\usepackage{geometry}

\usepackage{multirow}

\newcommand{\hongyang}[1]{\textcolor{black}{#1}}

\newcommand{\OPT}{\textup{\textsf{OPT}}}

\newcommand{\sign}{\textup{\textsf{sign}}}

\newcommand{\diag}{\textsf{Diag}}
\newcommand{\ber}{\textup{\textsf{Ber}}}

\newcommand{\Appendix}[1]{the full version for}

\newtheorem{theorem}{Theorem}[section]
\newtheorem{lemma}[theorem]{Lemma}
\newtheorem{proposition}[theorem]{Proposition}

\newtheorem{claim}{Claim}

\newtheorem{definition}{Definition}
\newtheorem{conjecture}{Conjecture}
\newtheorem{condition}{Condition}

\newcommand{\e}{\mathbf{e}}

\renewcommand{\u}{\mathbf{u}}

\newcommand{\x}{\mathbf{x}}
\newcommand{\y}{\mathbf{y}}

\newcommand{\A}{\mathbf{A}}
\newcommand{\B}{\mathbf{B}}
\newcommand{\C}{\mathbf{C}}
\newcommand{\D}{\mathbf{D}}
\newcommand{\E}{\mathbf{E}}
\newcommand{\F}{\mathbf{F}}
\newcommand{\G}{\mathbf{G}}
\renewcommand{\H}{\mathbf{H}}
\newcommand{\I}{\mathbf{I}}

\renewcommand{\L}{\mathbf{L}}
\newcommand{\M}{\mathbf{M}}

\newcommand{\R}{\mathbb{R}}
\renewcommand{\S}{\mathbf{S}}
\newcommand{\T}{\mathbf{T}}
\newcommand{\U}{\mathbf{U}}
\newcommand{\V}{\mathbf{V}}
\newcommand{\W}{\mathbf{W}}
\newcommand{\X}{\mathbf{X}}
\newcommand{\Y}{\mathbf{Y}}
\newcommand{\Z}{\mathbf{Z}}
\newcommand{\rank}{\textup{\textsf{rank}}}

\newcommand{\bLambda}{\mathbf{\Lambda}}

\newcommand{\0}{\mathbf{0}}

\renewcommand{\comment}[1]{}

\newcommand{\tr}{\textsf{tr}}
\newcommand{\cA}{\mathcal{A}}

\newcommand{\cC}{\mathcal{C}}
\newcommand{\cD}{\mathcal{D}}
\newcommand{\cG}{\mathcal{G}}

\newcommand{\cI}{\mathcal{I}}

\newcommand{\cS}{\mathcal{S}}
\newcommand{\cT}{\mathcal{T}}
\newcommand{\cU}{\mathcal{U}}
\newcommand{\cV}{\mathcal{V}}
\newcommand{\cN}{\mathcal{N}}
\newcommand{\cO}{\mathcal{O}}
\newcommand{\cP}{\mathcal{P}}
\newcommand{\cQ}{\mathcal{Q}}

\DeclareMathOperator*{\argmax}{argmax}
\DeclareMathOperator*{\argmin}{argmin}

\newenvironment{proofoutline}{\noindent{\emph{Proof Sketch. }}}{\hfill$\square$\medskip}

\title{Matrix Completion and Related Problems via Strong Duality}
\usepackage{times}

\author{Maria-Florina Balcan\thanks{Carnegie Mellon University. Email: ninamf@cs.cmu.edu} \and Yingyu Liang\thanks{University of Wisconsin-Madison. Email: yliang@cs.wisc.edu } \and David P. Woodruff\thanks{Carnegie Mellon University. Email: dwoodruf@cs.cmu.edu}    \and
Hongyang Zhang\thanks{Corresponding author.  Carnegie Mellon University. Email: hongyanz@cs.cmu.edu}
}

\date{}

\begin{document}

\maketitle

\begin{abstract}
This work studies the \emph{strong duality of non-convex matrix factorization problems}: we show that under certain dual conditions, these problems and its dual have the same optimum. This has been well understood for convex optimization, but little was known for non-convex problems. We propose a novel analytical framework and show that under certain dual conditions, the optimal solution of the matrix factorization program is the same as its bi-dual and thus the global optimality of the non-convex program can be achieved by solving its bi-dual which is convex. These dual conditions are satisfied by a wide class of matrix factorization problems, although matrix factorization problems are hard to solve in full generality.
This analytical framework may be of independent interest to non-convex optimization more broadly.

We apply our framework to two prototypical matrix factorization problems: matrix completion and robust Principal Component Analysis (PCA). These are examples of efficiently recovering a hidden matrix given limited reliable observations of it.
Our framework shows that exact recoverability and strong duality hold with nearly-optimal sample complexity guarantees for matrix completion and robust PCA.
\end{abstract}

\thispagestyle{empty}
\setcounter{page}{0}
\newpage

\section{Introduction}
\vspace{+0.2cm}
Non-convex matrix factorization problems have been an emerging object of study in theoretical computer science~\cite{jain2013low,hardt2014understanding,sun2015guaranteed,razenshteyn2016weighted}, optimization~\cite{wen2012solving,shen2014augmented}, machine learning~\cite{bhojanapalli2016global,ge2016matrix,ge2015escaping,jain2010guaranteed,li2016recovery,ICML2012Wang_233}, and many other domains. In theoretical computer science and optimization, the study of such models has led to significant advances in provable algorithms that converge to local minima in linear time~\cite{jain2013low,hardt2014understanding,sun2015guaranteed,agarwal2016finding,allen2016katyusha}. In machine learning, matrix factorization serves as a building block for large-scale prediction and recommendation systems, e.g., the winning submission for the Netflix prize~\cite{koren2009matrix}. Two prototypical examples are matrix completion and robust Principal Component Analysis (PCA).

This work develops a novel framework to analyze a class of non-convex matrix factorization problems with strong duality, which leads to exact recoverability for matrix completion and robust Principal Component Analysis (PCA) via the solution to a convex problem.
The matrix factorization problems can be stated as finding a target matrix $\X^*$ in the form of $\X^*=\A\B$, by minimizing the objective function $H(\A\B)+\frac{1}{2}\|\A\B\|_F^2$ over factor matrices $\A\in\R^{n_1\times r}$ and $\B\in\R^{r\times n_2}$ with a known value of $r\ll \min\{n_1,n_2\}$, where $H(\cdot)$ is some function that characterizes the desired properties of $\X^*$.

\vspace{+0.15cm}
Our work is motivated by several promising areas where our analytical framework for non-convex matrix factorizations is applicable. The first area is low-rank matrix completion, where it has been shown that a low-rank matrix can be exactly recovered by finding a solution of the form $\A\B$ that is consistent with the observed entries (assuming that it is incoherent)~\cite{jain2013low,sun2015guaranteed,ge2016matrix}. This problem has received a tremendous amount of attention due to its important role in optimization and its wide applicability in many areas such as quantum information theory and collaborative filtering~\cite{hardt2014understanding,zhang2016completing,balcan2016noise}. The second area is robust PCA, a fundamental problem of interest in data processing that aims at recovering both the low-rank and the sparse components exactly from their superposition~\cite{Candes,netrapalli2014non,gu2016low,Zhang2015AAAI,zhang2016completing,yi2016fast}, where the low-rank component corresponds to the product of $\A$ and $\B$ while the sparse component is captured by a proper choice of function $H(\cdot)$, e.g., the $\ell_1$ norm~\cite{Candes,awasthi2016learning}. We believe our analytical framework can be potentially applied to other non-convex problems more broadly, e.g., matrix sensing~\cite{tu2015low}, dictionary learning~\cite{sun2016complete}, weighted low-rank approximation~\cite{razenshteyn2016weighted,li2016recovery}, and deep linear neural network~\cite{kawaguchi2016deep}, which may be of independent interest.

\vspace{+0.15cm}
Without assumptions on the structure of the objective function, direct formulations of matrix factorization problems are NP-hard to optimize in general~\cite{hardt2014computational,Zhang:Counterexample}. With standard assumptions on the structure of the problem and with sufficiently many samples, these optimization problems can be solved efficiently, e.g., by convex relaxation~\cite{Candes2009exact,chen2015incoherence}. Some other methods run local search algorithms given an initialization close enough to the global solution in the basin of attraction~\cite{jain2013low,hardt2014understanding,sun2015guaranteed,ge2015escaping,jin2017escape}. However, these methods have sample complexity significantly larger than the information-theoretic lower bound; see Table \ref{table: comparison of sample complexity on matrix completion} for a comparison. The problem becomes more challenging when the number of samples is small enough that the sample-based initialization is far from the desired solution, in which case the algorithm can run into a local minimum or a saddle point.


\vspace{+0.15cm}
Another line of work has focused on studying the loss surface of matrix factorization problems, providing positive results for approximately achieving global optimality.
One nice property in this line of research is that there is no spurious local minima for specific applications such as matrix completion~\cite{ge2016matrix}, matrix sensing~\cite{bhojanapalli2016global}, dictionary learning~\cite{sun2016complete}, phase retrieval~\cite{sun2016geometric}, linear deep neural networks~\cite{kawaguchi2016deep}, etc.
However, these results are based on concrete forms of objective functions. Also, even when any local minimum is guaranteed to be globally optimal, in general it remains NP-hard to escape high-order saddle points~\cite{anandkumar2016efficient}, and additional arguments are needed to show the achievement of a local minimum.
Most importantly, all existing results rely on strong assumptions on the sample size.

\subsection{Our Results}
\vspace{-0.1cm}
Our work studies the exact recoverability problem for a variety of non-convex
matrix factorization problems.
The goal is to provide a unified framework to analyze a large class of matrix factorization problems, and to achieve efficient algorithms. Our main results show that although matrix factorization problems are hard to optimize in general, \emph{under certain dual conditions the duality gap is zero}, and thus the problem can be converted to an equivalent convex program. The main theorem of our framework is the following.

\medskip
\noindent\textbf{Theorems~\ref{theorem: strong duality with condition (c)} (Strong Duality. Informal).}
\emph{Under certain dual conditions, strong duality holds for the non-convex optimization problem
\vspace{-0.3cm}
\begin{equation}
\label{equ: informal objective function}
(\widetilde\A,\widetilde\B)=\argmin_{\A\in\R^{n_1\times r},\B\in\R^{r\times n_2}} F(\A,\B)=H(\A\B)+\frac{1}{2}\|\A\B\|_F^2,\quad H(\cdot)\text{\ is convex and closed},
\end{equation}
where ``the function $H(\cdot)$ is closed'' means that for each $\alpha\in\R$, the sub-level set $\{\X\in\R^{n_1\times n_2}:H(\X)\le\alpha\}$ is a closed set.
In other words, problem \eqref{equ: informal objective function} and its bi-dual problem
\begin{equation}
\label{equ: informal bi-dual objective function}
\widetilde\X=\argmin_{\X\in\R^{n_1\times n_2}} H(\X)+\|\X\|_{r*},\vspace{-0.3cm}
\end{equation}
have exactly the same optimal solutions in the sense that $\widetilde\A\widetilde\B=\widetilde\X$, where $\|\X\|_{r*}$ is a convex function defined by $\|\X\|_{r*}=\max_\M \langle\M,\X\rangle-\frac{1}{2}\|\M\|_r^2$ and $\|\M\|_r^2=\sum_{i=1}^r\sigma_i^2(\M)$ is the sum of the first $r$ largest squared singular values.
}

\medskip
Theorem \ref{theorem: strong duality with condition (c)} connects the non-convex program \eqref{equ: informal objective function} to its convex counterpart via strong duality; see Figure \ref{figure: bridge}. We mention that strong duality rarely happens in the non-convex optimization region: low-rank matrix approximation~\cite{overton1992sum} and quadratic optimization with two quadratic constraints~\cite{beck2006strong} are among the few paradigms that enjoy such a nice property.
Given strong duality, the computational issues of the original problem can be overcome by solving the convex bi-dual problem \eqref{equ: informal bi-dual objective function}.

The positive result of our framework is complemented by a lower bound to formalize the hardness of the above problem in general. Assuming that the random 4-SAT problem is hard (see Conjecture \ref{conj:r4sat})~\cite{razenshteyn2016weighted}, we give a strong negative result for deterministic algorithms. If also {\sf BPP = P} (see Section \ref{section: computations} for a discussion), then the same conclusion holds for randomized algorithms succeeding with probability at least $2/3$.

\begin{figure}[tb]
\centering\vspace{-0.5cm}
\includegraphics[width=1\textwidth]{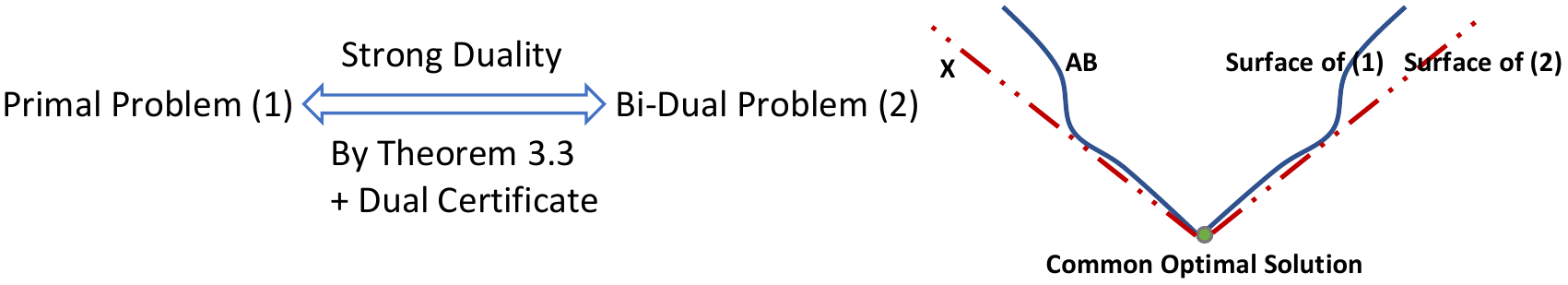}\vspace{-0.2cm}
\caption{Strong duality of matrix factorizations.}
\label{figure: bridge}\vspace{-0.4cm}
\end{figure}

\medskip
\noindent\textbf{Theorem~\ref{theorem: lower bound for general optimization} (Hardness Statement. Informal).}
\emph{Assuming that random 4-SAT is hard on average, there is a problem in the form of \eqref{equ: informal objective function} such that any deterministic algorithm achieving $(1+\epsilon)\OPT$ in the objective function value with $\epsilon\le \epsilon_0$ requires $2^{\Omega(n_1+n_2)}$ time, where $\textup{\textsf{OPT}}$ is the optimum and $\epsilon_0>0$ is an absolute constant. If {\sf BPP = P}, then the same conclusion holds for randomized algorithms succeeding with probability at least $2/3$.
}

Our framework only requires the dual conditions in Theorem \ref{theorem: strong duality with condition (c)} to be verified. We will show that two prototypical problems, matrix completion and robust PCA, obey the conditions.
They belong to the linear inverse problems of form \eqref{equ: informal objective function} with a proper choice of function $H(\cdot)$, which aim at exactly recovering a hidden matrix {$\X^*$ with $\rank(\X^*)\le r$} given a limited number of linear observations of it.

For matrix completion, the linear measurements are of the form $\{\X_{ij}^*: (i,j) \in \Omega\}$, where $\Omega$ is the support set which is uniformly distributed among all subsets of $[n_1] \times [n_2]$ of cardinality $m$. With strong duality, we can either study the exact recoverability of the primal problem \eqref{equ: informal objective function}, or investigate the validity of its convex dual (or bi-dual) problem \eqref{equ: informal bi-dual objective function}. Here we study the former with tools from geometric analysis.
Recall that in the analysis of matrix completion, one typically requires an $\mu$-incoherence condition \hongyang{for a given rank-$r$ matrix $\X^*$ with skinny SVD $\U\mathbf{\Sigma}\V^T$~\cite{recht2011simpler,Candes2010power}:}
\begin{align}
\|\U^T\e_i\|_2\le &\sqrt{\frac{\mu r}{n_1}},\qquad\mbox{and}\qquad \|\V^T\e_i\|_2\le \sqrt{\frac{\mu r}{n_2}},\qquad \text{for all }i\label{equ: incoherence}
\end{align}
where $\e_i$'s are vectors with $i$-th entry equal to $1$ and other entries equal to $0$. The incoherence condition claims that information spreads throughout the left and right singular vectors and is quite standard in the matrix completion literature. Under this standard condition, we have the following results.

\medskip
\noindent\textbf{Theorems~\ref{theorem: uniqueness of matrix completion}, \ref{theorem: matrix completion}, and \ref{theorem: lower bound for matrix completion, strong incoherence} (Matrix Completion. Informal).}
\emph{$\X^*\in\R^{n_1\times n_2}$ is the unique matrix of rank at most $r$ that is consistent with the $m$ measurements with minimum Frobenius norm by a high probability, provided that $m=\cO(\kappa^2\mu(n_1+n_2)r\log (n_1+n_2)\log_{2\kappa}(n_1+n_2))$ and $\X^*$ satisfies incoherence \eqref{equ: incoherence}. In addition, there exists a convex optimization for matrix completion in the form of \eqref{equ: informal bi-dual objective function} that exactly recovers $\X^*$ with high probability, provided that $m=\cO(\kappa^2\mu(n_1+n_2)r\log (n_1+n_2)\log_{2\kappa}(n_1+n_2))$, where $\kappa$ is the condition number of $\X^*$.}

\begin{table}[tb]
\caption{Comparison of matrix completion methods. Here $\kappa=\sigma_1(\X^*)/\sigma_r(\X^*)$ is the condition number of $\X^*\in\R^{n_1\times n_2}$, $\epsilon$ is the accuracy such that the output $\widetilde\X$ obeys $\|\widetilde\X-\X^*\|_F\le\epsilon$, $n_{(1)}=\max\{n_1,n_2\}$ and $n_{(2)}=\min\{n_1,n_2\}$.}
\label{table: comparison of sample complexity on matrix completion}
\centering
\begin{tabular}{c|cc}%
\hline
Work & Sample Complexity & $\mu$-Incoherence\\
\hline\hline
\cite{jain2013low} & $\cO\left(\kappa^4\mu^2r^{4.5}n_{(1)}\log n_{(1)}\log\left(\frac{r\|\X^*\|_F}{\epsilon}\right)\right)$ & Condition \eqref{equ: incoherence}\\
\hline
\cite{hardt2014understanding} & $\cO\left(\mu rn_{(1)}(r+\log\left(\frac{n_{(1)}\|\X^*\|_F}{\epsilon}\right)\frac{\|\X^*\|_F^2}{\sigma_r^2}\right)$ & Condition \eqref{equ: incoherence}\\
\hline
\cite{ge2016matrix} & $\cO(\max\{\mu^6\kappa^{16}r^4,\mu^4\kappa^4r^6\}n_{(1)}\log^2 n_{(1)})$ & $\|\X^*_{i:}\|_2\le \frac{\mu}{\sqrt{n_{(2)}}}\|\X^*\|_F$\\
\hline
\cite{sun2015guaranteed} & $\cO(rn_{(1)}\kappa^2\max\left\{\mu \log n_{(2)},\sqrt{\frac{n_{(1)}}{n_{(2)}}}\mu^2r^6\kappa^4\right\}$ & Condition \eqref{equ: incoherence}\\
\hline
\cite{zheng2016convergence} & $\cO(\mu r^2n_{(1)}\kappa^2\max(\mu,\log n_{(1)}))$ & Condition \eqref{equ: incoherence}\\
\hline
\cite{gamarnik2017matrix} & $\cO\left(\left(\mu^2r^4\kappa^2+\mu r\log\left(\frac{\|\X^*\|_F}{\epsilon}\right)\right)n_{(1)}\log\left(\frac{\|\X^*\|_F}{\epsilon}\right)\right)$ & Condition \eqref{equ: incoherence}\\
\hline
\cite{zhao2015nonconvex} & $\cO\left(\mu r^3n_{(1)}\log n_{(1)}\log\left(\frac{1}{\epsilon}\right)\right)$ & Condition \eqref{equ: incoherence}\\
\hline
\cite{keshavan2010matrix} & $\cO\left(n_{(2)}r\sqrt{\frac{n_{(1)}}{n_{(2)}}}\kappa^2\max\left\{\mu\log n_{(2)},\mu^2r\sqrt{\frac{n_{(1)}}{n_{(2)}}}\kappa^4\right\}\right)$ & Similar to \eqref{equ: incoherence} and \eqref{equ: strong incoherence for RPCA}\\
\hline
\cite{Gross2011recovering} & $\cO(\mu r n_{(1)}\log^2 n_{(1)})$ & Conditions \eqref{equ: incoherence} and \eqref{equ: strong incoherence for RPCA}\\
\hline
\cite{chen2015incoherence} & $\cO(\mu r n_{(1)}\log^2 n_{(1)})$ & Condition \eqref{equ: incoherence}\\
\hline
Ours & $\cO(\kappa^2\mu r n_{(1)}\log (n_{(1)})\log_{2\kappa}(n_{(1)}))$ & Condition \eqref{equ: incoherence}\\
\hline
\hline
Lower Bound\tablefootnote{This lower bound is information-theoretic.} \cite{Candes2010power} & $\Omega(\mu r n_{(1)}\log n_{(1)})$ & Condition \eqref{equ: incoherence}\\
\hline
\end{tabular}
\vspace{-0.3cm}
\end{table}

\medskip
To the best of our knowledge, our result is the first to connect convex matrix completion to non-convex matrix completion, two parallel lines of research that have received significant attention in the past few years. Table \ref{table: comparison of sample complexity on matrix completion} compares our result with prior results. 

For robust PCA, instead of studying exact recoverability of problem \eqref{equ: informal objective function} as for matrix completion, we investigate problem \eqref{equ: informal bi-dual objective function} directly. The robust PCA problem is to decompose a given matrix $\D=\X^*+\S^*$ into the sum of a low-rank component $\X^*$ and a sparse component $\S^*$~\cite{agarwal2012noisy}. We obtain the following theorem for robust PCA.

\medskip
\noindent\textbf{Theorems~\ref{theorem: robust PCA} (Robust PCA. Informal).}
\emph{There exists a convex optimization formulation for robust PCA in the form of problem \eqref{equ: informal bi-dual objective function} that exactly recovers the incoherent matrix $\X^*\in\R^{n_1\times n_2}$ and $\S^*\in\R^{n_1\times n_2}$ with high probability, even if
$\rank(\X^*)=\Theta\left(\frac{\min\{n_1,n_2\}}{\mu\log^2 \max\{n_1,n_2\}}\right)$ and the size of the support of $\S^*$ is $m=\Theta(n_1n_2)$, where the support set of $\S^*$ is uniformly distributed among all sets of cardinality $m$, and the incoherence parameter $\mu$ satisfies constraints \eqref{equ: incoherence} and $\|\X^*\|_\infty\le\sqrt{\frac{\mu r}{n_1n_2}}\sigma_r(\X^*)$.
}

\medskip
The bounds in Theorem \ref{theorem: robust PCA} match the best known results in the robust PCA literature when the supports of $\S^*$ are uniformly sampled~\cite{Candes}, while our assumption is arguably more intuitive; see Section~\ref{sec:rpca}. Note that our results hold even when $\X^*$ is close to full rank and a constant fraction of the entries have noise.
Independently of our work, Ge et al.~\cite{Rong2017No} developed a framework to analyze the loss surface of low-rank problems, and applied the framework to matrix completion and robust PCA. Their bounds are: for matrix completion, the sample complexity is $\cO(\kappa^6\mu^4r^6(n_1+n_2)\log (n_1+n_2))$; for robust PCA, the outlier entries are deterministic and the number that the method can tolerate is $\cO\left(\frac{n_1n_2}{\mu r\kappa^5}\right)$. Zhang et al.~\cite{zhang2017nonconvex} also studied the robust PCA problem using non-convex optimization, where the outlier entries are deterministic and the number of outliers that their algorithm can tolerate is $\cO\left(\frac{n_1n_2}{r\kappa}\right)$. 
The strong duality approach is unique to our work.

\subsection{Our Techniques}
\label{section: Our Techniques}
%

\medskip
\noindent{\textbf{Reduction to Low-Rank Approximation.}}
Our results are inspired by the low-rank approximation problem:
\begin{equation}
\label{equ: PCA}
\min_{\A\in\R^{n_1\times r},\B\in\R^{r\times n_2}} \frac{1}{2}\|\hongyang{-\widetilde\bLambda}-\A\B\|_F^2.
\end{equation}
We know that all local solutions of \eqref{equ: PCA} are globally optimal (see Lemma \ref{lemma: local-global}) and that strong duality holds for any given matrix $-\widetilde\bLambda\in\R^{n_1\times n_2}$~\cite{grussler2016low}. To extend this property to our more general problem \eqref{equ: informal objective function}, our main insight is to reduce problem \eqref{equ: informal objective function} to the form of \eqref{equ: PCA} using the $\ell_2$-regularization term. While some prior work attempted to apply a similar reduction, their conclusions either \hongyang{depended} on unrealistic conditions on local solutions, e.g., all local solutions are rank-deficient~\cite{haeffele2014structured,grussler2016low}, or their conclusions relied on strong assumptions on the objective functions, e.g., that the objective functions are twice-differentiable~\cite{haeffele2015global}. Instead, our general results formulate strong duality via the existence of a dual certificate $\widetilde\bLambda$. For concrete applications, the existence of a dual certificate is then converted to mild assumptions, e.g., that the number of measurements is sufficiently large and the positions of measurements are randomly distributed. We will illustrate the importance of randomness below.

\comment{
More specifically, denote by $(\widetilde\A,\widetilde\B)$ the optimal solution to \eqref{equ: informal objective function}. Define $\cT=\{\widetilde\A\L+\M\widetilde\B:\ \L\in\R^{r\times n_2},\M\in\R^{n_1\times r}\}$, $\cT^\perp$ the complement of $\cT$, and $\cP_\cT$ the orthogonal projection onto subspace $\cT$. Let $\partial H(\X)=\{\bLambda\in\R^{n_1\times n_2}: H(\Y)\ge H(\X)+\langle\bLambda,\Y-\X\rangle\mbox{ for any } \Y\}$ be the sub-differential of function $H$ evaluated at $\X$. To perform the reduction from problem \eqref{equ: informal objective function} to \eqref{equ: PCA}, we study the Lagrangian $L(\A,\B,\bLambda)$ of \eqref{equ: informal objective function}, which is equivalent to problem \eqref{equ: PCA} if we fix $\bLambda=\widetilde\bLambda$. We show that, for a fixed Lagrangian multiplier $\widetilde\bLambda\in\partial H(\widetilde\A\widetilde\B)$, minimizing the primal problem \eqref{equ: informal objective function} reduces to minimizing the Lagrangian function $L(\A,\B,\widetilde\bLambda)$, (i.e., problem \eqref{equ: PCA}, thus strong duality holds), if $(\widetilde\A,\widetilde\B)$ remains globally optimal to $L(\A,\B,\widetilde\bLambda)$ (i.e., problem \eqref{equ: PCA}). This can be translated to: a) $\exists\widetilde\bLambda\in\partial H(\widetilde\A\widetilde\B)$, b) $(\widetilde\A,\widetilde\B)$ is a stationary point of the Lagrangian $L(\A,\B,\widetilde\bLambda)$ so that c) $\widetilde\A\widetilde\B=\textsf{svd}_r(-\widetilde\bLambda)$, where $\textsf{svd}_r(-\widetilde\bLambda)=\U_{:,1:r}\mathbf{\Sigma}_{1:r,1:r}\V_{:,1:r}^T$ if $\U\mathbf{\Sigma}\V^T$ is the SVD of $-\widetilde\bLambda$, and $\U_{:,1:r}$ and $\mathbf{\Sigma}_{1:r,1:r}$ are the first $r$ columns of $\U$ and top left $r\times r$ submatrix of $\mathbf{\Sigma}$, respectively. We note that conditions b) and c) can be rephrased as $\cP_\cT(-\widetilde\bLambda)=\widetilde\A\widetilde\B$ and $\sigma_1(\cP_{\cT^\perp}\widetilde\bLambda)<\sigma_r(\widetilde\A\widetilde\B)$, respectively. To satisfy conditions a), b) and c) simultaneously, one may want to find a certificate $\widetilde\bLambda$ such that among all matrices $\bLambda\in\partial H(\widetilde\A\widetilde\B)$ \hongyang{(i.e., condition a))} with $\cP_\cT(-\bLambda)=\widetilde\A\widetilde\B$ \hongyang{(i.e., condition b))}, $\widetilde\bLambda$ is the one with minimum Frobenius norm, so that condition c) is easier to satisfy. Following this principle, we build our dual certificate $\widetilde\bLambda$ by $-\cP_{\partial H}\cP_{\cT}(\cP_\cT\cP_{\partial H}\cP_\cT)^{-1}(\widetilde\A\widetilde\B)$. It can be easily checked that conditions a) and b) hold for our construction. Thus the remainder is to prove condition c) for specific applications. We observe that $\partial H=\Omega$ for matrix completion, where $\Omega$ is the linear space that characterizes the sample support, i.e., $\Omega=\{\X:\X_{ij}=0 \text{ if } (i,j) \text{ is unsampled}\}$.\footnote{Without confusion, we denote by $\Omega$ both the linear subspace and the index set concerning the sampled positions.} This nice property serves as a bridge, connecting our analytical framework to the concrete application of matrix completion. We will then state how to prove condition c) by randomness as follows.
}

\medskip
\noindent{\textbf{The Blessing of Randomness.}}
The desired dual certificate $\widetilde\bLambda$ may not exist in the deterministic world. A hardness result~\cite{razenshteyn2016weighted} shows that for the problem of weighted low-rank approximation, which can be cast in the form of \eqref{equ: informal objective function}, without some randomization in the
measurements made on the underlying low rank matrix,
it is NP-hard to achieve a good objective value, not to mention to achieve strong duality.
A similar phenomenon was observed for deterministic matrix completion~\cite{hardt2012algorithms}. Thus we should utilize such randomness to analyze the existence of a dual certificate. For matrix completion, the assumption that the measurements are random is standard, under which, the angle between the space $\Omega$ (the space of matrices which are consistent with observations) and the space $\cT$ (the space of matrices which are low-rank) is small with high probability, namely, $\X^*$ is almost the unique low-rank matrix that is consistent with the measurements. Thus, our dual certificate can be represented as another form of a convergent Neumann series concerning the projection operators on the spaces $\Omega$ and $\cT$. The remainder of the proof is to show that such a construction obeys the dual conditions.

To prove the dual conditions for matrix completion, we use the fact that the subspace $\Omega$ and the complement space $\cT^\perp$ are almost orthogonal when the sample size is sufficiently large. This implies the projection of our dual certificate on the space $\cT^\perp$ has a very small norm, which exactly matches the dual conditions.

\begin{wrapfigure}{R}{5cm}
\includegraphics[width=5cm]{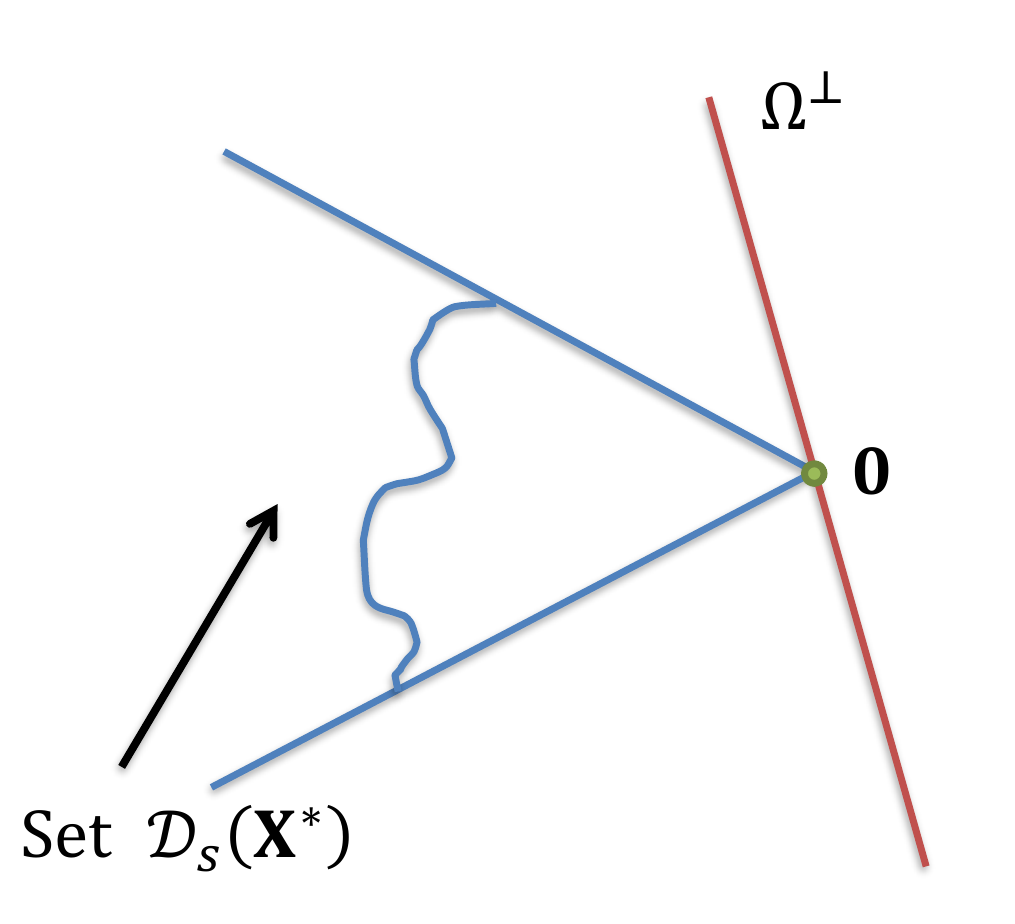}
\vspace{-0.6cm}
\caption{Feasibility.}
\label{figure: exact recovery}
\end{wrapfigure}
\medskip
\noindent{\textbf{Non-Convex Geometric Analysis.}}
Strong duality implies that the primal problem \eqref{equ: informal objective function} and its bi-dual problem \eqref{equ: informal bi-dual objective function} have exactly the same solutions in the sense that $\widetilde\A\widetilde\B=\widetilde\X$. Thus, to show exact recoverability of linear inverse problems such as matrix completion and robust PCA, it suffices to study either the non-convex primal problem \eqref{equ: informal objective function} or its convex counterpart \eqref{equ: informal bi-dual objective function}. Here we do the former analysis for matrix completion. We mention that traditional techniques~\cite{Candes2010power,recht2011simpler,chandrasekaran2012convex} for convex optimization break down for our non-convex problem, since the subgradient of a non-convex objective function may not even exist~\cite{boyd2004convex}. Instead, we apply tools from geometric analysis~\cite{Vershynin2010lectures} to analyze the geometry of problem \eqref{equ: informal objective function}. Our non-convex geometric analysis is in stark contrast to prior techniques of convex geometric analysis~\cite{Vershynin2014Estimation} where convex combinations of non-convex constraints were used to define the Minkowski functional (e.g., in the definition of atomic norm) while our method uses the non-convex constraint itself.

For matrix completion, problem \eqref{equ: informal objective function} has two hard constraints: a) the rank of the output matrix should be no larger than $r$, as implied by the form of $\A\B$; b) the output matrix should be consistent with the sampled measurements, i.e., $\cP_\Omega(\A\B)=\cP_\Omega(\X^*)$. We study the feasibility condition of problem \eqref{equ: informal objective function} from a geometric perspective: $\widetilde\A\widetilde\B=\X^*$ is the unique optimal solution to problem \eqref{equ: informal objective function} if and only if starting from $\X^*$, either the rank of $\X^*+\D$ or $\|\X^*+\D\|_F$ increases for all directions $\D$'s in the constraint set $\Omega^\perp=\{\D\in\R^{n_1\times n_2}:\cP_\Omega(\X^*+\D)=\cP_\Omega(\X^*)\}$. This can be geometrically interpreted as the requirement that the set $\cD_\cS(\X^*)=\{\X-\X^*\in\R^{n_1\times n_2}:\rank(\X)\le r, \|\X\|_F \le \|\X^*\|_F\}$ and the constraint set $\Omega^\perp$ must intersect uniquely at $\0$ (see Figure \ref{figure: exact recovery}). This can then be shown by a dual certificate argument.

%

\comment{
\medskip
\noindent{\textbf{Bypassing the Golfing Scheme to Obtain Optimal Bounds.}}
Compared to the nuclear norm method~\cite{chen2015incoherence,recht2011simpler,Gross2011recovering}, our non-convex geometric analysis leads to a multiplicative $\log n_{(1)}$ factor improvement in the sample complexity for matrix completion. The main reason is that we can avoid using the technique in the so-called golfing scheme, which is loose but inevitable in the analysis of the nuclear norm method. The golfing scheme is a sequential way of building a certificate vector: suppose we would like to build a dual certificate that is close to $\W_*$. Starting from an initialization $\W_0$, in the $i$-th step the golfing scheme calculates the distance between our current guess $\W_i$ and the target $\W_*$, and takes fresh i.i.d. samples the entries of $\W_i-\W_*$ projected onto subspace $\cT$ in each step, using them for our next guess of $\W_{i+1}-\W_*$. In expectation, the distance between $\W_i$ and $\W_*$ becomes smaller and after $2\lceil\log n_{(1)}\rceil$ steps, $\W_{2\lceil\log n_{(1)}\rceil}$ is very close to $\W_*$ with high probability. However, the golfing scheme re-samples matrices in each of $2\lceil\log n_{(1)}\rceil$ iterative steps, and so the analysis incurs an additional multiplicative $\log n_{(1)}$ factor in the sample complexity. Although it might be possible to modify the golfing scheme to use the same samples in each iteration or to reuse many samples, it is unknown how to do so and this remains an important open question. In contrast, our analysis completely avoids the necessity of resampling by going through a very different geometric approach.
}

\medskip
\noindent{\textbf{Putting Things Together.}} We summarize our new analytical framework with the following figure.
\begin{figure}[h]
\centering
\includegraphics[width=1\textwidth]{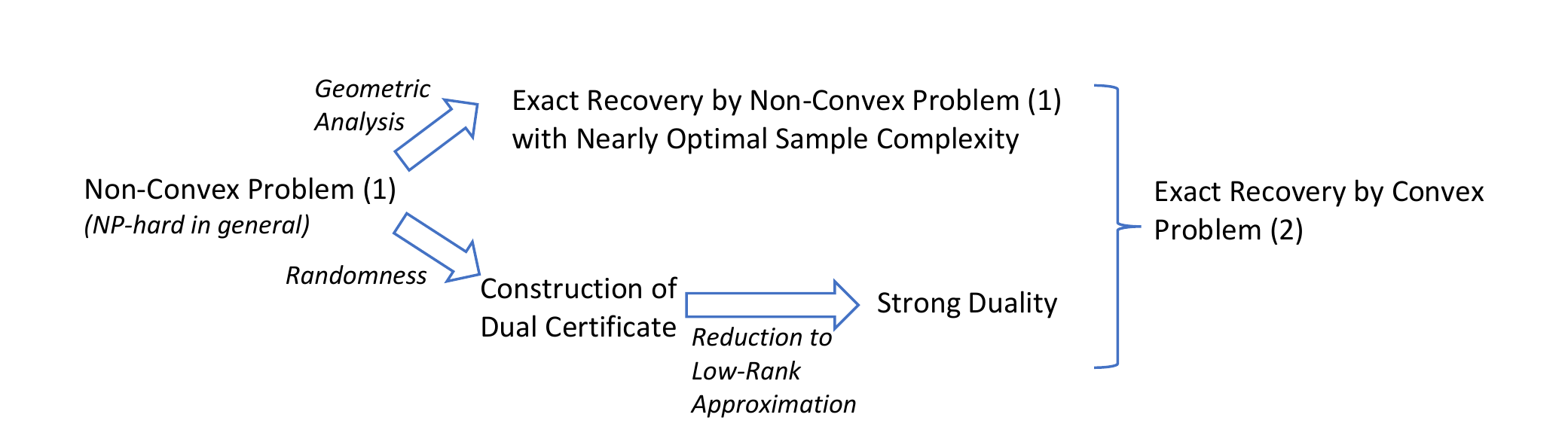}
\end{figure}

\vspace{-0.2cm}

\medskip
\noindent{\textbf{Other Techniques.}} An alternative method is to investigate the exact recoverability of problem \eqref{equ: informal bi-dual objective function} via standard convex analysis. We find that the sub-differential of our induced function $\|\cdot\|_{r*}$ is very similar to that of the nuclear norm. With this observation, we prove the validity of robust PCA in the form of \eqref{equ: informal bi-dual objective function} by combining this property of $\|\cdot\|_{r*}$ with standard techniques from \cite{Candes}.

\section{Preliminaries}
\label{section: Notations and Preliminaries}

We will use calligraphy to represent a set, bold capital letters to represent a matrix, bold lower-case letters to represent a vector, and lower-case letters to represent scalars. Specifically, we denote by $\X^*\in\R^{n_1\times n_2}$ the underlying matrix. We use $\X_{:t}\in\R^{n_1\times 1}$ ($\X_{t:}\in\R^{1\times n_2}$) to indicate the $t$-th column (row) of $\X$. The entry in the $i$-th row, $j$-th column of $\X$ is represented by $\X_{ij}$. The condition number of $\X$ is $\kappa=\sigma_1(\X)/\sigma_r(\X)$. We let $n_{(1)}=\max\{n_1, n_2\}$ and $n_{(2)}=\min\{n_1, n_2\}$.
For a function $H(\M)$ on an input matrix $\M$, its conjugate function $H^*$ is defined by $H^*(\bLambda)=\max_{\M}\langle\bLambda,\M\rangle-H(\M)$. Furthermore, let $H^{**}$ denote the conjugate function of $H^*$.

We will frequently use $\rank(\X)\le r$ to constrain the rank of $\X$. This can be equivalently represented as $\X=\A\B$, by restricting the number of columns of $\A$ and rows of $\B$ to be $r$. For norms, we denote by $\|\X\|_F=\sqrt{\sum_{ij}\X_{ij}^2}$ the Frobenius norm of matrix $\X$. Let $\sigma_1(\X)\ge\sigma_2(\X)\ge...\ge\sigma_r(\X)$ be the non-zero singular values of $\X$. The nuclear norm (a.k.a. trace norm) of $\X$ is defined by $\|\X\|_*=\sum_{i=1}^r\sigma_i(\X)$, and the operator norm of $\X$ is $\|\X\|=\sigma_1(\X)$. Denote by $\|\X\|_\infty=\max_{ij}|\X_{ij}|$. For two matrices $\A$ and $\B$ of equal dimensions, we denote by $\langle \A,\B\rangle=\sum_{ij}\A_{ij}\B_{ij}$. We denote by $\partial H(\X)=\{\bLambda\in\R^{n_1\times n_2}: H(\Y)\ge H(\X)+\langle\bLambda,\Y-\X\rangle\mbox{ for any } \Y\}$ the sub-differential of function $H$ evaluated at $\X$. We define the indicator function of convex set $\cC$ by
$\I_\cC(\X)=
\begin{cases}
0, & \mbox{if } \X\in\cC;\\
+\infty, & \mbox{otherwise}.
\end{cases}$
For any non-empty set $\cC$, denote by $\textsf{cone}(\cC)=\{t\X:\X\in\cC,\ t\ge 0\}$.

We denote by $\Omega$ the set of indices of observed entries, and $\Omega^\perp$ its complement. Without confusion, $\Omega$ also indicates the linear subspace formed by matrices with entries in $\Omega^\perp$ being $0$. We denote by $\cP_\Omega: \R^{n_1\times n_2}\rightarrow\R^{n_1\times n_2}$ the orthogonal projector of subspace $\Omega$. We will consider a single norm for these operators, namely, the operator norm denoted by $\|\cA\|$ and defined by $\|\cA\|=\sup_{\|\X\|_F=1}\|\cA(\X)\|_F$. For any orthogonal projection operator $\cP_\cT$ to any subspace $\cT$, we know that $\|\cP_\cT\|=1$ whenever $\textup{\textsf{dim}}(\cT)\not=0$. For distributions, denote by $\cN(0,1)$ a standard Gaussian random variable, $\textsf{Uniform}(m)$ the uniform distribution of cardinality $m$, and $\ber(p)$ the Bernoulli distribution with success probability $p$.

\vspace{-0.2cm}
\section{$\ell_2$-Regularized Matrix Factorizations: A New Analytical Framework}
\vspace{-0.2cm}
\label{section: Framework}
In this section, we develop a novel framework to analyze a general class of $\ell_2$-regularized matrix factorization problems. Our framework can be applied to different specific problems and leads to nearly optimal sample complexity guarantees.
In particular, we study the $\ell_2$-regularized matrix factorization problem
\begin{equation*}
(\textbf{P})\qquad \min_{\A\in\R^{n_1\times r},\B\in\R^{r\times n_2}} F(\A,\B)=H(\A\B)+\frac{1}{2}\|\A\B\|_F^2,\ H(\cdot)\text{\ is convex and closed.}\vspace{-0.2cm}
\end{equation*}
We show that under suitable conditions the duality gap between (\textbf{P}) and its dual \hongyang{(bi-dual)} problem is zero, so problem (\textbf{P}) can be converted to an equivalent convex problem.

\comment{
\medskip
\noindent \textbf{Examples.} There are many examples in which the data under study can be modelled in the form of (\textbf{P}) and fit our framework. To shed light on the nature of our results, we give two examples inspired by contemporary challenges in recommendation system, statistics and machine learning.
\begin{itemize}[leftmargin=*]\setlength{\itemsep}{-\itemsep}
\item
\noindent \textbf{Matrix Completion.}
Setting $H(\A\B)=\I_{\{\M:\cP_\Omega(\M)=\cP_\Omega(\X^*)\}}(\A\B)$, we obtain the models for low-rank matrix completion:
$
\min_{\A\in\R^{n_1\times r},\B\in\R^{r\times n_2}} \frac{1}{2}\|\A\B\|_F^2,\ \mbox{s.t.}\  \cP_\Omega(\A\B)=\cP_\Omega(\X^*),
$
where $\Omega\sim \textsf{\mbox{Uniform}}(m)$. As shown in Section \ref{section: Matrix Completion},
the solution to this optimization problem exactly recovers the incoherent matrix $\X^*$ with optimal sample complexity up to a constant factor.
\item
\noindent \textbf{Matrix Sensing.} Setting $H(\A\B)=\I_{\{\M:\cA(\M)=\cA(\X^*)\}}(\A\B)$, we obtain the models for low-rank matrix sensing:
$
\min_{\A\in\R^{n_1\times r},\B\in\R^{r\times n_2}} \frac{1}{2}\|\A\B\|_F^2,\ \mbox{s.t.}\ \cA(\A\B)=\cA(\X^*),
$
where $\cA(\cdot)=\{\langle\A_i,\cdot\rangle\}_{i=1}^m$ is the linear operator with standard Gaussian matrix $\{\A_i\}_{i=1}^m$. As shown in Section \ref{section: Matrix Sensing}, the solution to this optimization problem performs as well as the best known results, being able to exactly recover the underlying matrix $\X^*$ with sample complexity as small as the information-theoretic limit.
\end{itemize}
}

\vspace{-0.2cm}
\subsection{Strong Duality}
\label{section: proof of strong duality}
\vspace{-0.1cm}
We first consider an easy case where $H(\A\B)=\frac{1}{2}\|\widehat \Y\|_F^2-\langle \widehat \Y,\A\B\rangle$ for a fixed $\widehat \Y$, leading to the objective function $\frac{1}{2}\|\widehat \Y-\A\B\|_F^2$. For this case, we establish the following lemma.
\vspace{-0.2cm}
\begin{lemma}
\label{lemma: local-global}
For any given matrix $\widehat \Y \in \R^{n_1 \times n_2}$, any local minimum of $f(\A,\B)=\frac{1}{2}\|\widehat \Y-\A\B\|_F^2$ over $\A\in\R^{n_1\times r}$ and $\B\in\R^{r\times n_2} (r \le \min\{n_1, n_2\})$ is globally optimal, given by $\textup{\textsf{svd}}_r(\widehat{\Y})$. The objective function $f(\A,\B)$ around any saddle point has a negative second-order directional curvature. Moreover, $f(\A,\B)$ has no local maximum.\footnote{Prior work studying the loss surface of low-rank matrix approximation assumes that the matrix $\widetilde\bLambda$ is of full rank and does not have the same singular values~\cite{Baldi1989neural}. In this work, we generalize this result by removing these two assumptions.}
\end{lemma}

The proof of Lemma \ref{lemma: local-global} is basically to calculate the gradient of $f(\A,\B)$ and let it equal to zero; see Appendix \ref{section: Proof of PCA Lemma} for details. Given this lemma, we can reduce $F(\A,\B)$ to the form $\frac{1}{2}\|\widehat{\Y}-\A\B\|_F^2$ for some $\widehat \Y$ plus an extra term:
\begin{equation}\vspace{-0.2cm}
\label{equ: Lagrangian form}
\begin{split}
F(\A,\B)&\hspace{-0.05cm}=\hspace{-0.05cm}\frac{1}{2}\|\A\B\|_F^2+H(\A\B)\hspace{-0.05cm}=\hspace{-0.05cm}\frac{1}{2}\|\A\B\|_F^2+H^{**}(\A\B)\hspace{-0.05cm}=\hspace{-0.05cm}\max_\bLambda\frac{1}{2}\|\A\B\|_F^2+\langle\bLambda,\A\B\rangle-H^*(\bLambda)\\
&=\max_\bLambda\frac{1}{2}\|-\bLambda-\A\B\|_F^2-\frac{1}{2}\|\bLambda\|_F^2-H^*(\bLambda)\triangleq\max_\bLambda L(\A,\B,\bLambda),
\end{split}
\end{equation}
where we define $L(\A,\B,\bLambda)\triangleq\frac{1}{2}\|-\bLambda-\A\B\|_F^2-\frac{1}{2}\|\bLambda\|_F^2-H^*(\bLambda)$ \hongyang{as the Lagrangian of problem (\textbf{P})},\footnote{One can easily check that $L(\A,\B,\bLambda)=\min_\M L'(\A,\B,\M,\bLambda)$, where $L'(\A,\B,\M,\bLambda)$ is the Lagrangian of the constraint optimization problem $\min_{\A,\B,\M} \frac{1}{2}\|\A\B\|_F^2+H(\M),\ \mbox{s.t.}\ \M=\A\B$. With a little abuse of notation, we call $L(\A,\B,\bLambda)$ the Lagrangian of the unconstrained problem (\textbf{P}) as well.} and the second equality holds because $H$ is closed and convex w.r.t. the argument $\A\B$.
For any fixed value of $\bLambda$, by Lemma \ref{lemma: local-global}, any local minimum of $L(\A,\B,\bLambda)$ is globally optimal, because minimizing $L(\A,\B,\bLambda)$ is equivalent to minimizing $\frac{1}{2}\|-\bLambda-\A\B\|_F^2$ for a fixed $\bLambda$.

The remaining part of our analysis is to choose a proper $\widetilde \bLambda$ such that $(\widetilde \A,\widetilde\B,\widetilde \bLambda)$ is a primal-dual saddle point of $L(\A,\B,\bLambda)$, \hongyang{so that $\min_{\A,\B} L(\A,\B,\widetilde\bLambda)$ and problem (\textbf{P}) have the same optimal solution $(\widetilde \A,\widetilde\B)$.}
For this, we introduce the following condition, \hongyang{and later we will show that the condition holds with high probability.}

\begin{condition}
\label{lemma: Lagrangian multiplier}
For a solution ($\widetilde \A$, $\widetilde \B$) to problem \textup{(\textbf{P})}, there exists an $\widetilde\bLambda\in\partial_{\X} H(\X)|_{\X=\widetilde\A\widetilde \B}$ such that\vspace{-0.2cm}
\begin{equation}
\label{equ: Lambda condition}
-\widetilde\A\widetilde \B\widetilde \B^T=\widetilde \bLambda\widetilde \B^T\quad \mbox{and}\quad \widetilde \A^T(-\widetilde\A\widetilde \B)=\widetilde \A^T\widetilde \bLambda.
\end{equation}
\end{condition}

\noindent
\textbf{Explanation of Condition \ref{lemma: Lagrangian multiplier}.} We note that
$\nabla_\A L(\A,\B,\bLambda)=\A\B\B^T+\bLambda\B^T\ \mbox{and}\
\nabla_\B L(\A,\B,\bLambda)=\A^T\A\B+\A^T\bLambda$ for a fixed $\bLambda$.
In particular, if we set $\bLambda$ to be the $\widetilde \bLambda$ in \eqref{equ: Lambda condition}, then $\nabla_\A L(\A,\widetilde \B,\widetilde \bLambda)|_{\A=\widetilde \A}=\0$ and $\nabla_\B L(\widetilde \A,\B,\widetilde \bLambda)|_{\B=\widetilde \B}=\0$. So Condition \ref{lemma: Lagrangian multiplier} implies that $(\widetilde\A,\widetilde\B)$ is either a saddle point or a local minimizer of $L(\A,\B,\widetilde\bLambda)$ as a function of $(\A,\B)$ for the fixed $\widetilde\bLambda$.

The following lemma states that if it is a local minimizer, then strong duality holds.

\begin{lemma}[Dual Certificate]
\label{lemma: f and k local minimum}
Let $(\widetilde\A,\widetilde\B)$ be a global minimizer of $F(\A,\B)$. If there exists a dual certificate $\widetilde\bLambda$ satisfying Condition \ref{lemma: Lagrangian multiplier} and the pair $(\widetilde\A,\widetilde\B)$ is a local minimizer of $L(\A,\B,\widetilde\bLambda)$ for the fixed $\widetilde\bLambda$, then strong duality holds. Moreover, we have the relation $\widetilde\A\widetilde\B=\textup{\textsf{svd}}_r(-\widetilde\bLambda)$.
\end{lemma}

\vspace{-0.2cm}
\begin{proofoutline}
By the assumption of the lemma, we can show that $(\widetilde\A,\widetilde\B,\widetilde\bLambda)$ is a primal-dual saddle point to the Lagrangian $L(\A,\B,\bLambda)$; see Appendix \ref{section: proof of strong duality}. To show strong duality, by the fact that $F(\A,\B)=\max_\bLambda L(\A,\B,\bLambda)$ and that $\widetilde\bLambda=\argmax_{\bLambda} L(\widetilde\A,\widetilde\B,\bLambda)$, we have
$
F(\widetilde\A,\widetilde\B)=L(\widetilde\A,\widetilde\B,\widetilde\bLambda)\le L(\A,\B,\widetilde\bLambda),
$
for any $\A,\B,$ where the inequality holds because $(\widetilde\A,\widetilde\B,\widetilde\bLambda)$ is a primal-dual saddle point of $L$. So on the one hand,
$
\min_{\A,\B}\max_\bLambda L(\A,\B,\bLambda)=F(\widetilde\A,\widetilde\B)\le \min_{\A,\B} L(\A,\B,\widetilde\bLambda)\le\max_\bLambda\min_{\A,\B} L(\A,\B,\bLambda).
$
On the other hand, by weak duality, we have
$
\min_{\A,\B}\max_\bLambda L(\A,\B,\bLambda)\ge \max_\bLambda\min_{\A,\B} L(\A,\B,\bLambda).
$
Therefore, $\min_{\A,\B}\max_\bLambda L(\A,\B,\bLambda)=\max_\bLambda\min_{\A,\B} L(\A,\B,\bLambda)$, i.e., strong duality holds. Therefore,
$
\widetilde\A\widetilde\B=\argmin_{\A\B} L(\A,\B,\widetilde\bLambda)=\argmin_{\A\B} \frac{1}{2}\|\A\B\|_F^2+\langle\widetilde\bLambda,\A\B\rangle-H^*(\widetilde\bLambda)=\argmin_{\A\B} \frac{1}{2}\|-\widetilde\bLambda-\A\B\|_F^2=\textsf{svd}_r(-\widetilde\bLambda),
$
as desired.
\end{proofoutline}

This lemma then leads to the following theorem.

\begin{theorem}
\label{theorem: strong duality with condition (c)}
Denote by $(\widetilde\A,\widetilde\B)$ the optimal solution of problem \textup{(\textbf{P})}. Define a matrix space
$$\cT\triangleq \{\widetilde\A\X^T+\Y\widetilde\B,\ \X\in\R^{n_2\times r},\ \Y\in\R^{n_1\times r}\}.$$
Then strong duality holds for problem \textup{(\textbf{P})}, provided that there exists $\widetilde \bLambda$ such that\vspace{-0.1cm}
\begin{flalign}
\label{equ: dual condition}
\begin{split}
\textup{(1)}\ \widetilde \bLambda\in\partial H(\widetilde\A\widetilde\B)\triangleq\Psi, \qquad \textup{(2)}\ \cP_\cT(-\widetilde\bLambda)=\widetilde\A\widetilde\B,\qquad \textup{(3)}\ \|\cP_{\cT^\perp}\widetilde\bLambda\|< \sigma_r(\widetilde\A\widetilde\B).\\
\end{split}
\end{flalign}
\end{theorem}\vspace{-0.3cm}

\begin{proof}
The proof idea is to construct a dual certificate $\widetilde\bLambda$ so that the conditions in Lemma \ref{lemma: f and k local minimum} hold.
$\widetilde\bLambda$ should satisfy the following:\vspace{-0.1cm}
\begin{flalign}
\label{equ: dual conditions for general problem}
\begin{split}
&\mbox{(a)}\qquad \widetilde \bLambda\in\partial H(\widetilde\A\widetilde\B),\qquad \mbox{(by Condition \ref{lemma: Lagrangian multiplier})}\\
&\mbox{(b)}\qquad (\widetilde\A\widetilde\B+\widetilde \bLambda)\widetilde \B^T=\0\quad \mbox{and}\quad \widetilde \A^T(\widetilde\A\widetilde\B+\widetilde \bLambda)=\0,\qquad \mbox{(by Condition \ref{lemma: Lagrangian multiplier})}\\
&\mbox{(c)}\qquad \widetilde\A\widetilde\B=\textsf{svd}_r(-\widetilde\bLambda).\qquad \mbox{(by the local minimizer assumption and Lemma \ref{lemma: local-global})}\\
\end{split}
\end{flalign}
It turns out that for any matrix $\M\in\R^{n_1\times n_2}$, $\cP_{\cT^\perp}\M=(\I-\widetilde\A\widetilde\A^\dag)\M(\I-\widetilde\B\widetilde\B^\dag)$ and so $\|\cP_{\cT^\perp}\M\|\le \|\M\|$, a fact that we will frequently use in the sequel. Denote by $\cU$ the left singular space of $\widetilde\A\widetilde\B$ and $\cV$ the right singular space. Then the linear space $\cT$ can be equivalently represented as $\cT=\cU+\cV$. Therefore, $\cT^\perp=(\cU+\cV)^\perp=\cU^\perp\cap\cV^\perp$. With this, we note that: (b) $(\widetilde\A\widetilde\B+\widetilde\bLambda)\widetilde\B^T=\0$ and $\widetilde\A^T(\widetilde\A\widetilde\B+\widetilde\bLambda)=\0$ imply $\widetilde\A\widetilde\B+\widetilde \bLambda\in\textsf{Null}(\widetilde\A^T)=\textsf{Col}(\widetilde\A)^\perp$ and $\widetilde\A\widetilde\B+\widetilde \bLambda\in\textsf{Row}(\widetilde\B)^\perp$ (so $\widetilde\A\widetilde\B+\widetilde \bLambda\in\cT^\perp$), and vice versa. And (c) $\widetilde\A\widetilde\B=\textsf{svd}_r(-\widetilde\bLambda)$ implies that for an orthogonal decomposition
$-\widetilde\bLambda=\widetilde\A\widetilde\B+\E, \mbox{ where } \widetilde\A\widetilde\B\in\cT, \mbox{ and }\E\in\cT^\perp$, we have
$\|\E\|<\sigma_r(\widetilde\A\widetilde\B)$. Conversely, $\|\E\|<\sigma_r(\widetilde\A\widetilde\B)$ and condition (b) imply $\widetilde\A\widetilde\B=\textup{\textsf{svd}}_r(-\widetilde\bLambda)$.
Therefore, the dual conditions in \eqref{equ: dual conditions for general problem} are equivalent to (1) $\widetilde \bLambda\in\partial H(\widetilde\A\widetilde\B)\triangleq \Psi$; (2) $\cP_\cT(-\widetilde\bLambda)=\widetilde\A\widetilde\B$; (3) $\|\cP_{\cT^\perp}\widetilde\bLambda\|< \sigma_r(\widetilde\A\widetilde\B)$.
\end{proof}

To show the dual condition in Theorem \ref{theorem: strong duality with condition (c)}, intuitively, we need to show that the angle $\theta$ between subspace $\cT$ and $\Psi$ is small (see Figure \ref{figure: dual conditions for general problem}) for a specific function $H(\cdot)$. In the following \hongyang{(see Section \ref{section: Dual Certification by Least Squares})}, we will demonstrate applications that, with randomness, obey this dual condition with high probability.

\begin{figure}[tb]
\centering\vspace{-0.6cm}
\includegraphics[width=0.5\textwidth]{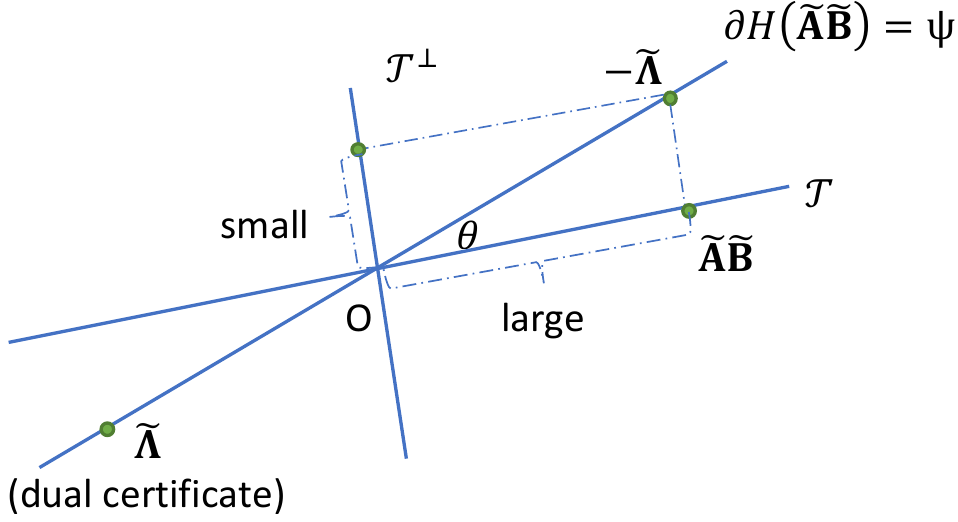}\vspace{-0.2cm}
\caption{Geometry of dual condition \eqref{equ: dual condition} for general matrix factorization problems.}
\label{figure: dual conditions for general problem}
\end{figure}

\vspace{-0.4cm}
\section{Matrix Completion}
\label{section: Matrix Completion}
\vspace{-0.1cm}

In matrix completion, there is a hidden matrix $\X^*\in\R^{n_1\times n_2}$ with rank $r$. We are given measurements $\{\X_{ij}^*: (i,j) \in \Omega\}$, where $\Omega\sim\textup{\textsf{Uniform}}(m)$, i.e., $\Omega$ is sampled uniformly at random from all subsets of $[n_1] \times [n_2]$ of cardinality $m$. The goal is to exactly recover $\X^*$ with high probability.
Here we apply our unified framework in Section \ref{section: Framework} to matrix completion, by setting $H(\cdot)=\I_{\{\M:\cP_\Omega(\M)=\cP_\Omega(\X^*)\}}(\cdot)$.

A quantity governing the difficulties of matrix completion is the incoherence parameter $\mu$. Intuitively, matrix completion is possible only if the information spreads evenly throughout the low-rank matrix. This intuition is captured by the incoherence conditions. Formally, denote by $\U\mathbf{\Sigma}\V^T$ the skinny SVD of a fixed $n_1\times n_2$ matrix $\X$ of rank $r$. Cand\`{e}s et al.~\cite{Candes,Candes2009exact,recht2011simpler,zhang2016completing} introduced the $\mu$-incoherence condition \eqref{equ: incoherence} to the low-rank matrix $\X$.
For conditions \eqref{equ: incoherence}, it can be shown that $1\le \mu\le \frac{n_{(1)}}{r}$. The condition holds for many random matrices with incoherence parameter $\mu$ about $\sqrt{r\log n_{(1)}}$~\cite{keshavan2010matrix}.

We first propose a non-convex optimization problem whose unique solution is indeed the ground truth $\X^*$, and then apply our framework to show that strong duality holds for this non-convex optimization and its bi-dual optimization problem.
\begin{theorem}[Uniqueness of Solution]
\label{theorem: uniqueness of matrix completion}
Let $\Omega\sim\textup{\textsf{Uniform}}(m)$ be the support set uniformly distributed among all sets of cardinality $m$. Suppose that $m\ge c\kappa^2\mu n_{(1)}r\log n_{(1)}\log_{2\kappa} n_{(1)}$ for an absolute constant $c$ and $\X^*$ obeys $\mu$-incoherence \eqref{equ: incoherence}. Then $\X^*$ is the unique solution of non-convex optimization
\begin{equation}
\label{equ: primal problem}
\min_{\A,\B} \frac{1}{2}\|\A\B\|_F^2,\quad \mathrm{s.t.}\quad \cP_\Omega(\A\B)=\cP_\Omega(\X^*),
\end{equation}
with probability at least $1-n_{(1)}^{-10}$.
\end{theorem}
\begin{proofoutline}
Here we sketch the proof and defer the details to Appendix \ref{section: proof of information-theoretic upper bound for matrix completion}.
We consider the feasibility of the matrix completion problem:
\begin{equation}
\label{equ: rank-constraint matrix completion main}
\mbox{Find a matrix $\X\in\R^{n_1\times n_2}$ such that}\quad\cP_\Omega(\X)=\cP_\Omega(\X^*),\quad\rank(\X)\le r, \quad \|\X\|_F \le  \|\X^*\|_F .
\end{equation}
Our proof first identifies a feasibility condition for problem \eqref{equ: rank-constraint matrix completion main}, and then shows that $\X^*$ is the only matrix which obeys this feasibility condition when the sample size is large enough. More specifically, we note that $\X^*$ obeys the conditions in problem \eqref{equ: rank-constraint matrix completion main}. Therefore, $\X^*$ is the only matrix which obeys condition \eqref{equ: rank-constraint matrix completion main} if and only if $\X^*+\D$ does not follow the condition for all $\D$, i.e., $\cD_\cS(\X^*)\cap \Omega^\perp=\{\0\}$, where $\cD_\cS(\X^*)$ is defined as
\begin{equation*}
\cD_\cS(\X^*)=\{\X-\X^*\in\R^{n_1\times n_2}:\rank(\X)\le r,\ \|\X\|_F\le \|\X^*\|_F\}.
\end{equation*}
This can be shown by combining the satisfiability of the dual conditions in Theorem~\label{theorem: strong duality with condition (c)}, and the well known fact that $\cT\cap\Omega^\perp=\{\0\}$ when the sample size is large.
\end{proofoutline}


Given the non-convex problem, we are ready to state our main theorem for matrix completion.
\begin{theorem}[Efficient Matrix Completion]
\label{theorem: matrix completion}
Let $\Omega\sim\textup{\textsf{Uniform}}(m)$ be the support set uniformly distributed among all sets of cardinality $m$.
Suppose $\X^*$ has condition number $\kappa=\sigma_1(\X^*)/\sigma_r(\X^*)$.
Then there are absolute constants $c$ and $c_0$ such that with probability at least $1-c_0n_{(1)}^{-10}$, the output of the convex problem
\begin{equation}
\label{equ: equality constraint matrix completion}
\widetilde\X=\argmin_{\X} \|\X\|_{r*},\quad\textup{\mbox{s.t.}}\quad \cP_\Omega(\X)=\cP_\Omega(\X^*),
\end{equation}
is unique and exact, i.e., $\widetilde\X=\X^*$, provided that $m\ge c\kappa^2\mu rn_{(1)}\log_{2\kappa} (n_{(1)})\log (n_{(1)})$ and $\X^*$ obeys $\mu$-incoherence \eqref{equ: incoherence}. Namely, strong duality holds for problem \eqref{equ: primal problem}.\footnote{In addition to our main results on strong duality, in a previous version of this paper we also claimed a tight information-theoretic bound on the number of samples required for matrix completion; the proof of that latter claim was problematic as stated, and so we have removed that claim in this version.}
\end{theorem}

\begin{proofoutline}
We have shown in Theorem \ref{theorem: uniqueness of matrix completion} that the problem
$
(\widetilde\A,\widetilde\B)=\argmin_{\A,\B} \frac{1}{2}\|\A\B\|_F^2,\textup{\mbox{s.t.}}\cP_\Omega(\A\B)=\cP_\Omega(\X^*),
$
exactly recovers $\X^*$, i.e., $\widetilde\A\widetilde\B=\X^*$, with small sample complexity. So if strong duality holds, this non-convex optimization problem can be equivalently converted to the convex program \eqref{equ: equality constraint matrix completion}. Then Theorem \ref{theorem: matrix completion} is straightforward from strong duality.

It now suffices to apply our unified framework in Section \ref{section: Framework} to prove the strong duality. We show that the dual condition in Theorem \ref{theorem: strong duality with condition (c)} holds with high probability by the following arguments.
Let $(\widetilde\A,\widetilde\B)$ be a global solution to problem \eqref{equ: equality constraint matrix completion}.
For $H(\X)=\I_{\{\M\in\R^{n_1\times n_2}:\ \cP_\Omega\M=\cP_\Omega\X^*\}}(\X)$, we have
\begin{equation*}
\begin{split}
\Psi &= \partial H(\widetilde\A\widetilde\B)
=\{\G\in\R^{n_1\times n_2}: \langle\G,\widetilde\A\widetilde\B\rangle\ge\langle \G,\Y\rangle,\ \mbox{for any } \Y\in\R^{n_1\times n_2}\mbox{ s.t. } \cP_\Omega\Y=\cP_\Omega\X^*\}\\
&=\{\G\in\R^{n_1\times n_2}: \langle\G,\X^*\rangle\ge\langle \G,\Y\rangle,\ \mbox{for any } \Y\in\R^{n_1\times n_2}\mbox{ s.t. } \cP_\Omega\Y=\cP_\Omega\X^*\}
=\Omega,
\end{split}
\end{equation*}
where the third equality holds since $\widetilde\A\widetilde\B=\X^*$. Then we only need to show
\begin{flalign}
\label{equ: dual condition for matrix completion main}
\begin{split}
\mbox{(1)}\ \widetilde \bLambda\in\Omega,\qquad \mbox{(2)}\  \cP_{\cT}(-\widetilde\bLambda)=\widetilde\A\widetilde\B,\qquad \mbox{(3)}\  \|\cP_{\cT^\perp}\widetilde\bLambda\|<\frac{2}{3}\sigma_r(\widetilde\A\widetilde\B).\\
\end{split}
\end{flalign}
It is interesting to see that dual condition \eqref{equ: dual condition for matrix completion main} can be satisfied if the angle $\theta$ between subspace $\Omega$ and subspace $\cT$ is very small; see Figure \ref{figure: dual conditions for general problem}. When the sample size $|\Omega|$ becomes larger and larger, the angle $\theta$ becomes smaller and smaller (e.g., when $|\Omega|=n_1n_2$, the angle $\theta$ is zero as $\Omega=\R^{n_1\times n_2}$). We show that the sample size $m= \Omega(\kappa^2\mu rn_{(1)}\log_{2\kappa} (n_{(1)})\log (n_{(1)}))$ is a sufficient condition for condition \eqref{equ: dual condition for matrix completion main} to hold.
\end{proofoutline}

This positive result matches a lower bound from prior work up to a logarithmic factor, which shows that the sample complexity in Theorem \ref{theorem: uniqueness of matrix completion} is nearly optimal.
\begin{theorem}[Information-Theoretic Lower Bound. \cite{Candes2010power}, Theorem 1.7]
\label{theorem: lower bound for matrix completion, strong incoherence}
Denote by $\Omega\sim\textup{\textsf{Uniform}}(m)$ the support set uniformly distributed among all sets of cardinality $m$. Suppose that $m\le c\mu n_{(1)}r\log n_{(1)}$ for an absolute constant $c$. Then there exist infinitely many $n_1\times n_2$ matrices $\X'$ of rank at most $r$ obeying $\mu$-incoherence \eqref{equ: incoherence} such that $\cP_\Omega(\X')=\cP_\Omega(\X^*)$, with probability at least $1-n_{(1)}^{-10}$.
\end{theorem}

\section{Robust Principal Component Analysis}\label{sec:rpca}
In this section, we develop our theory for robust PCA based on our framework. In the problem of robust PCA, we are given an observed matrix of the form $\D=\X^*+\S^*$, where $\X^*$ is the ground-truth matrix and $\S^*$ is the corruption matrix which is sparse. The goal is to recover the hidden matrices $\X^*$ and $\S^*$ from the observation $\D$. We set $H(\X)=\lambda\|\D-\X\|_1$.

To make the information \hongyang{spreads} evenly throughout the matrix, the matrix cannot have one entry whose absolute value is significantly larger than other entries.
For the robust PCA problem, Cand\`{e}s et al.~\cite{Candes} introduced an extra incoherence condition (Recall that $\X^*=\U\mathbf{\Sigma}\V^T$ is the skinny SVD of $\X^*$)
\begin{align}
\|\U\V^T\|_\infty\le\sqrt\frac{\mu r}{n_1n_2}.\label{equ: incoherence 2}
\end{align}
In this work, we make the following incoherence assumption for robust PCA instead of \eqref{equ: incoherence 2}:
\begin{equation}
\label{equ: strong incoherence for RPCA}
\|\X^*\|_\infty\le\sqrt{\frac{\mu r}{n_1n_2}}\sigma_r(\X^*).
\end{equation}
Note that condition \eqref{equ: strong incoherence for RPCA} is very similar to the incoherence condition \eqref{equ: incoherence 2} for the robust PCA problem, but the two notions are incomparable. Note that condition \eqref{equ: strong incoherence for RPCA} has an intuitive explanation, namely, that the entries must scatter almost uniformly across the low-rank matrix.

We have the following results for robust PCA.

\begin{theorem}[Robust PCA]
\label{theorem: robust PCA}
Suppose $\X^*$ is an $n_1\times n_2$ matrix of rank $r$, and obeys incoherence \eqref{equ: incoherence} and \eqref{equ: strong incoherence for RPCA}.
Assume that the support set $\Omega$ of $\S^*$ is uniformly distributed among all sets of cardinality $m$. Then with probability at least $1-cn_{(1)}^{-10}$, the output of the optimization problem
\begin{equation*}
(\widetilde\X,\widetilde\S)=\argmin_{\X,\S} \|\X\|_{r*}+\lambda \|\S\|_1,\quad\textup{\mbox{s.t.}}\quad \D=\X+\S,
\end{equation*}
with $\lambda=\frac{\sigma_r(\X^*)}{\sqrt{n_{(1)}}}$ is exact, namely, $\widetilde\X=\X^*$ and $\widetilde\S=\S^*$, provided that
$\rank(\X^*)\le \rho_r\frac{n_{(2)}}{\mu \log^2 n_{(1)}}\ \mbox{and}\ m\le\rho_sn_1n_2$,
where $c$, $\rho_r$, and $\rho_s$ are all positive absolute constants, and function $\|\cdot\|_{r*}$ is given by \eqref{equ: r* norm}.
\end{theorem}
The bounds on the rank of $\X^*$ and the sparsity of $\S^*$ in Theorem \ref{theorem: robust PCA} match the best known results for robust PCA in prior work when we assume the support set of $\S^*$ is sampled uniformly~\cite{Candes}.

\section{Computational Aspects}

\label{section: computations}

\medskip
\noindent{\textbf{Computational Efficiency.}}
We discuss our computational efficiency given that we have strong duality. We note that the dual and bi-dual of primal problem (\textbf{P}) are given by (see Appendix \ref{section: Dual and Bi-Dual Problems})
\vspace{-0.3cm}
\begin{equation}
\label{equ: r* norm}
\begin{split}
&(\textbf{\textup{Dual, D1}})\qquad \max_{\bLambda\in\R^{n_1\times n_2}} -H^*(\bLambda)-\frac{1}{2}\|\bLambda\|_r^2,\quad \mbox{where } \|\bLambda\|_r^2=\sum_{i=1}^r\sigma_i^2(\bLambda),\\
&(\textbf{\textup{Bi-Dual, D2}})\qquad \min_{\M\in\R^{n_1\times n_2}} H(\M)+\|\M\|_{r*},\quad \mbox{where } \|\M\|_{r*}=\max_\X \langle\M,\X\rangle-\frac{1}{2}\|\X\|_r^2.
\end{split}
\end{equation}
Problems (\textbf{D1}) and (\textbf{D2}) can be solved efficiently due to their convexity. In particular, Grussler et al.~\cite{grussler2016low} provided a computationally efficient algorithm to compute the proximal operators of functions $\frac{1}{2}\|\cdot\|_r^2$ and $\|\cdot\|_{r*}$. Hence, the Douglas-Rachford algorithm can find global minimum up to an $\epsilon$ error in function value in time $\textsf{poly}(1/\epsilon)$~\cite{he20121}.

\medskip
\noindent{\textbf{Computational Lower Bounds.}}
Unfortunately, strong duality does not always hold for general non-convex problems (\textbf{P}). Here we present a very strong
lower bound based on the random 4-SAT hypothesis.
This is by now a fairly standard conjecture in complexity theory
\cite{f02} and gives us constant factor inapproximability of problem
(\textbf{P}) for
deterministic algorithms, even those running in exponential time.

If we additionally assume that {\sf BPP =  P}, where
{\sf BPP}
is the class of problems which can be solved in probabilistic polynomial
time, and {\sf P} is the class of problems which can be solved in
deterministic polynomial
time, then the same conclusion holds for randomized algorithms. This is
also a standard conjecture in complexity theory, as it is implied by
the existence of certain strong pseudorandom generators or if any problem
in deterministic exponential time has exponential size circuits
\cite{iw97}. Therefore, any subexponential time algorithm achieving
a sufficiently small constant
factor approximation to problem (\textbf{P}) in general
would imply a major breakthrough in complexity theory.

The lower bound is proved by a reduction from the Maximum Edge Biclique problem~\cite{ambuhl2011inapproximability}. The details are presented in Appendix \ref{section: proof of lower bounds}.

\begin{theorem}[Computational Lower Bound] \label{theorem: lower bound for general optimization}
Assume Conjecture~\ref{conj:r4sat} (the hardness of Random 4-SAT). Then there exists an absolute constant $\epsilon_0 > 0$ for which any deterministic algorithm achieving $(1+\epsilon) \OPT$ in the objective function value for problem \textup{(\textbf{P})} with $\epsilon \le \epsilon_0$, requires $2^{\Omega(n_1 + n_2)}$ time, where $\OPT$ is the optimum. If in addition, {\sf BPP = P}, then the same conclusion holds for randomized algorithms succeeding with probability at least $2/3$.
\end{theorem}

\begin{proofoutline}
Theorem~\ref{theorem: lower bound for general optimization} is proved by using the hypothesis that random 4-SAT is hard to show
  hardness of the Maximum Edge Biclique problem for deterministic algorithms. We then do a reduction from the Maximum Edge Biclique problem to our problem.
\end{proofoutline}

The complete proofs of other theorems/lemmas and related work can be found in the appendices.

\medskip
\noindent{\textbf{Acknowledgments.}} We thank Rina Foygel Barber, Rong Ge, Jason D. Lee, Zhouchen Lin, Guangcan Liu, Tengyu Ma, Benjamin Recht, Xingyu Xie, and Tuo Zhao for useful discussions. This work was supported in part by NSF grants NSF CCF-1422910, NSF CCF-1535967, NSF CCF-1451177, NSF IIS-1618714, NSF CCF-1527371, a Sloan Research Fellowship, a Microsoft Research Faculty Fellowship, DMS-1317308, Simons Investigator Award, Simons Collaboration Grant, and ONR-N00014-16-1-2329.

\newpage
\appendices

\newcommand{\poly}{{\mathrm{poly}}}

\section{Other Related Work}
\label{sec:related}
Non-convex matrix factorization is a popular topic studied in theoretical computer science~\cite{jain2013low,hardt2014understanding,sun2015guaranteed,razenshteyn2016weighted}, machine learning~\cite{bhojanapalli2016global,ge2016matrix,ge2015escaping,jain2010guaranteed,li2016recovery}, and optimization~\cite{wen2012solving,shen2014augmented}.
We review several lines of research on studying the global optimality of such optimization problems.

\medskip
\noindent{\textbf{Global Optimality of Matrix Factorization.}}
While lots of matrix factorization problems have been shown to have no spurious local minima, they either require additional conditions on the local minima, or are based on particular forms of the objective function. Specifically, Burer and Monteiro~\cite{burer2005local} showed that one can minimize $F(\A\A^T)$ for any convex function $F$ by solving for $\A$ directly without introducing any local minima, provided that the rank of the output $\A$ is larger than the rank of the true minimizer $\X_{true}$. However, such a condition is often impossible to check as $\rank(\X_{true})$ is typically unknown a priori. To resolve the issue, Bach et al.~\cite{bach2008convex} and Journ{\'e}e et al.~\cite{journee2010low} proved that $\X=\A\A^T$ is a global minimizer of $F(\X)$, if $\A$ is a rank-deficient local minimizer of $F(\A\A^T)$ and $F(\X)$ is a twice differentiable convex function. Haeffele and Vidal~\cite{haeffele2015global} further extended this result by allowing a more general form of objective function $F(\X)=G(\X)+H(\X)$, where $G$ is a twice differentiable convex function with compact level set and $H$ is a proper convex function such that $F$ is lower semi-continuous. However, a major drawback of this line of research is that these result fails when the local minimizer is of full rank.

\medskip
\noindent{\textbf{Matrix Completion.}}
Matrix completion is a prototypical example of matrix factorization.
One line of work on matrix completion builds on convex relaxation (e.g.,~\cite{srebro2005rank,Candes2009exact,Candes2010power,recht2011simpler,chandrasekaran2012convex,negahban2012restricted}).
Recently, Ge et al.~\cite{ge2016matrix} showed that matrix completion has no spurious local optimum, when $|\Omega|$ is sufficiently large and the matrix $\Y$ is incoherent. The result is only for positive semi-definite matrices and their sample complexity is not nearly optimal.

Another line of work is built upon good initialization for global convergence. Recent attempts showed that one can first compute some form of initialization (e.g., by singular value decomposition) that is close to the global minimizer and then use non-convex approaches to reach global optimality, such as alternating minimization, block coordinate descent, and gradient descent~\cite{keshavan2010matrixnoise,keshavan2010matrix,jain2013low,keshavan2012efficient,hardt2014understanding,bhojanapalli2015dropping,zheng2015convergent,zhao2015nonconvex,tu2015low,chen2015fast,sun2015guaranteed}.
In our result, in contrast, we can reformulate non-convex matrix completion problems as equivalent convex programs, which guarantees global convergence from any initialization.

\medskip
\noindent{\textbf{Robust PCA.}}
Robust PCA is also a prototypical example of matrix factorization. The goal is to recover both the low-rank and the sparse components exactly from their superposition~\cite{Candes,netrapalli2014non,gu2016low,Zhang2015AAAI,zhang2016completing,yi2016fast}. It has been widely applied to various tasks, such as video denoising, background
modeling, image alignment,
photometric stereo, texture
representation, subspace clustering, and spectral clustering.

There are typically two settings in the robust PCA literature: a) the support set of the sparse matrix is uniformly sampled~\cite{Candes,zhang2016completing}; b) the support set of the sparse matrix is deterministic, but the non-zero entries in each row or column of the matrix cannot be too large~\cite{yi2016fast,Rong2017No}. In this work, we discuss the first case. Our framework provides results that match the best known work in setting (b)~\cite{Candes}.

\medskip
\noindent{\textbf{Other Matrix Factorization Problems.}}
Matrix sensing is another typical matrix factorization problem~\cite{chandrasekaran2012convex,jain2013low,zhao2015nonconvex}. Bhojanapalli et al.~\cite{bhojanapalli2016global} and Tu et al.~\cite{tu2015low} showed that the matrix recovery model
$\min_{\A,\B}\frac{1}{2}\|\cA(\A\B-\Y)\|_F^2$,
achieves optimality for every local minimum, if the operator $\cA$ satisfies the restricted isometry property. They further gave a lower bound and showed that the unstructured operator $\cA$ may easily lead to a local minimum which is not globally optimal.

Some other matrix factorization problems are also shown to have nice geometric properties such as the property that all local minima are global minima. Examples include dictionary learning~\cite{sun2016complete}, phase retrieval~\cite{sun2016geometric}, and linear deep neural networks~\cite{kawaguchi2016deep}.
In multi-layer linear neural networks where the goal is to learn a multi-linear projection $\X^*=\prod_i \W_i$, each $\W_i$ represents the weight matrix that connects the hidden units in the $i$-th and $(i+1)$-th layers. The study of such linear models is central to the theoretical understanding of the loss surface of deep neural networks with non-linear activation functions~\cite{kawaguchi2016deep,choromanska2015loss}. In dictionary learning, we aim to recover a complete (i.e., square and invertible) dictionary matrix $\A$ from a given signal $\X$ in the form of $\X=\A\B$, provided that the representation coefficient $\B$ is sufficiently sparse. This problem centers around solving a non-convex matrix factorization problem with a sparsity constraint on the representation coefficient $\B$~\cite{bach2008convex,sun2016complete,sun2015completeII,arora2014provable}. Other high-impact examples of matrix factorization models range from the classic unsupervised learning problems like PCA, independent component analysis, and clustering, to the more recent problems such as non-negative matrix factorization, weighted low-rank matrix approximation, sparse coding, tensor decomposition~\cite{bhaskara2014smoothed,anandkumar2014tensor}, subspace clustering~\cite{zhang2015relations,Zhang:RobustLatLRR}, etc.
Applying our framework to these other problems is left for future work.

\medskip
\noindent{\textbf{Atomic Norms.}} The atomic norm is a recently proposed function for linear inverse problems~\cite{chandrasekaran2012convex}. Many well-known norms, e.g., the $\ell_1$ norm and the nuclear norm, serve as special cases of atomic norms. It has been widely applied to the problems of compressed sensing~\cite{tang2013compressed}, low-rank matrix recovery~\cite{candes2013simple}, blind deconvolution~\cite{ahmed2014blind}, etc. The norm is defined by the Minkowski functional associated with the convex hull of a set $\cA$: $\|\X\|_\cA=\inf\{t>0:\ \X\in t\cA\}$. In particular, if we set $\cA$ to be the convex hull of the infinite set of unit-$\ell_2$-norm rank-one matrices, then $\|\cdot\|_\cA$ equals to the nuclear norm. We mention that our objective term $\|\A\B\|_F$ in problem \eqref{equ: informal objective function} is similar to the atomic norm, but with slight differences: unlike the atomic norm, we set $\cA$ to be the infinite set of unit-$\ell_2$-norm rank-$r$ matrices for $\rank(\X)\le r$. With this, we achieve better sample complexity guarantees than the atomic-norm based methods.

\section{Proof of Lemma \ref{lemma: local-global}}
\label{section: Proof of PCA Lemma}

\noindent{\textbf{Lemma \ref{lemma: local-global}} (Restated)\textbf{.}}
\emph{For any given matrix $\widehat \Y \in \R^{n_1 \times n_2}$, any local minimum of $f(\A,\B)=\frac{1}{2}\|\widehat \Y-\A\B\|_F^2$ over $\A\in\R^{n_1\times r}$ and $\B\in\R^{r\times n_2} (r \le \min\{n_1, n_2\})$ is globally optimal, given by $\textup{\textsf{svd}}_r(\widehat{\Y})$. The objective function $f(\A,\B)$ around any saddle point has a negative second-order directional curvature. Moreover, $f(\A,\B)$ has no local maximum.
}

\begin{proof}
($\A$,$\B$) is a critical point of $f(\A,\B)$ if and only if $\nabla_\A f(\A,\B)=\0$ and $\nabla_\B f(\A,\B)=\0$, or equivalently,
\begin{equation}
\label{equ: gradient=0}
\A\B\B^T=\widehat{\Y}\B^T\quad \mbox{and}\quad \A^T\A\B=\A^T\widehat{\Y}.
\end{equation}
Note that for any fixed matrix $\A$ (resp. $\B$), the function $f(\A,\B)$ is convex in the coefficients of $\B$ (resp. $\A$).

To prove the desired lemma, we have the following claim.
\begin{claim}
If two matrices $\A$ and $\B$ define a critical point of $f(\A,\B)$, then the global mapping $\M=\A\B$ is of the form
\begin{equation*}
\M=\cP_{\A}\widehat{\Y},
\end{equation*}
with $\A$ satisfying
\begin{equation}
\label{equ: condition on A}
\A\A^\dag\widehat{\Y}\widehat{\Y}^T=\A\A^\dag\widehat{\Y}\widehat{\Y}^T\A\A^\dag=\widehat{\Y}\widehat{\Y}^T\A\A^\dag.
\end{equation}
\end{claim}

\begin{proof}
If $\A$ and $\B$ define a critical point of $f(\A,\B)$, then \eqref{equ: gradient=0} holds and the general solution to \eqref{equ: gradient=0} satisfies
\begin{equation}
\label{equ: B}
\B=(\A^T\A)^\dag\A^T\widehat{\Y}+(\I-\A^\dag\A)\L,
\end{equation}
for some matrix $\L$.
So $\M=\A\B=\A(\A^T\A)^\dag\A^T\widehat{\Y}=\A\A^\dag\widehat{\Y}=\cP_\A\widehat{\Y}$ by the property of the Moore-Penrose pseudo-inverse: $\A^\dag=(\A^T\A)^\dag\A^T$.

By \eqref{equ: gradient=0}, we also have
\begin{equation*}
\label{equ: MM^T}
\A\B\B^T\A^T=\widehat{\Y}\B^T\A^T\quad \mbox{or equivalently}\quad \M\M^T=\widehat{\Y}\M^T.
\end{equation*}
Plugging in the relation $\M=\A\A^\dag\widehat{\Y}$, \eqref{equ: MM^T} can be rewritten as
\begin{equation*}
\A\A^\dag\widehat{\Y}\widehat{\Y}^T\A\A^\dag=\widehat{\Y}\widehat{\Y}^T\A\A^\dag.
\end{equation*}
Note that the matrix $\A\A^\dag\widehat{\Y}\widehat{\Y}^T\A\A^\dag$ is symmetric. Thus
\begin{equation*}
\A\A^\dag\widehat{\Y}\widehat{\Y}^T\A\A^\dag = \A\A^\dag\widehat{\Y}\widehat{\Y}^T,
\end{equation*}
as desired.
\end{proof}

To prove Lemma \ref{lemma: local-global}, we also need the following claim.
\begin{claim}
Denote by $\cI=\{i_1,i_2,...,i_r\}$ any ordered $r$-index set (ordered by $\lambda_{i_j}, j\in[r]$ from the largest to the smallest) and $\lambda_i$, $i\in[n_1]$, the ordered eigenvalues of $\widehat{\Y}\widehat{\Y}^T\in\R^{n_1\times n_1}$ with $p$ distinct values. Let $\U=[\u_1,\u_2,...,\u_{n_1}]$ denote the matrix formed by the orthonormal eigenvectors of $\widehat{\Y}\widehat{\Y}^T\in\R^{n_1\times n_1}$ associated with the ordered $p$ eigenvalues, whose multiplicities are $m_1,m_2,...,m_p$ ($m_1+m_2+...+m_p=n_1$). For any matrix $\M$, let $\M_{:\cI}$ denote the submatrix $[\M_{i_1},\M_{i_2},...,\M_{i_r}]$ associated with the index set $\cI$.

Then two matrices $\A$ and $\B$ define a critical point of $f(\A,\B)$ if and only if there exists an ordered $r$-index set $\cI$, an invertible matrix $\C$, and an $r\times n$ matrix $\L$ such that
\begin{equation}
\label{equ: critical point}
\A=(\U\D)_{:\cI}\C\quad\mbox{and}\quad \B=\A^\dag\widehat{\Y}+(\I-\A^\dag\A)\L,
\end{equation}
where $\D$ is a $p$-block-diagonal matrix with each block corresponding to the eigenspace of an eigenvalue. For such a critical point, we have
\begin{equation*}
\A\B=\cP_\A\widehat{\Y},
\end{equation*}
\begin{equation}
\label{equ: objective function in spectral}
f(\A,\B)=\frac{1}{2}\left(\textup{\tr}(\widehat{\Y}\widehat{\Y}^T)-\sum_{i\in\cI}\lambda_i\right)=\frac{1}{2}\sum_{i\not\in\cI}\lambda_i.
\end{equation}
\end{claim}

\begin{proof}
Note that $\widehat{\Y}\widehat{\Y}^T$ is a real symmetric covariance matrix. So it can always be represented as $\U\bLambda\U^T$, where $\U\in\R^{n_1\times n_1}$ is an orthonormal matrix consisting of eigenvectors of $\widehat{\Y}\widehat{\Y}^T$ and $\bLambda\in\R^{n_1\times n_1}$ is a diagonal matrix with non-increasing eigenvalues of $\widehat{\Y}\widehat{\Y}^T$.

If $\A$ and $\B$ satisfy \eqref{equ: critical point} for some $\C$, $\L$, and $\cI$, then
\begin{equation*}
\A\B\B^T=\widehat{\Y}\B^T\quad \mbox{and}\quad \A^T\A\B=\A^T\widehat{\Y},
\end{equation*}
which is \eqref{equ: gradient=0}. So $\A$ and $\B$ define a critical point of $f(\A,\B)$.

For the converse, notice that
\begin{equation*}
\cP_{\U^T\A}=\U^T\A(\U^T\A)^\dag=\U^T\A\A^\dag\U=\U^T\cP_{\A}\U,
\end{equation*}
or equivalently, $\cP_{\A}=\U\cP_{\U^T\A}\U^T$. Thus \eqref{equ: condition on A} yields
\begin{equation*}
\U\cP_{\U^T\A}\U^T\U\bLambda\U^T=\U\bLambda\U^T\U\cP_{\U^T\A}\U^T,
\end{equation*}
or equivalently, $\cP_{\U^T\A}\bLambda=\bLambda\cP_{\U^T\A}$. Notice that $\bLambda\in\R^{n_1\times n_1}$ is a diagonal matrix with $p$ distinct eigenvalues of $\widehat{\Y}\widehat{\Y}^T$, and $\cP_{\U^T\A}$ is an orthogonal projector of rank $r$. So $\cP_{\U^T\A}$ is a rank-$r$ restriction of the block-diagonal PSD matrix $\F:=\T\T^T$ with $p$ blocks, each of which is an orthogonal projector of dimension $m_i$, corresponding to the eigenvalues $\lambda_i$, $i\in[p]$. Therefore, there exists an $r$-index set $\cI$ and a block-diagonal matrix $\D$ such that $\cP_{\U^T\A}=\D_{:\cI}\D_{:\cI}^T$, where $\D=\T$. It follows that \begin{equation*}
\cP_\A=\U\cP_{\U^T\A}\U^T=\U\D_{:\cI}\D_{:\cI}^T\U^T=(\U\D)_{:\cI}(\U\D)_{:\cI}^T.
\end{equation*}
Since the column space of $\A$ coincides with the column space of $(\U\D)_{:\cI}$, $\A$ is of the form $\A=(\U\D)_{:\cI}\C$, and $\B$ is given by \eqref{equ: B}. Thus $\A\B=\A(\A^T\A)^\dag\A^T\widehat{\Y}+\A(\I-\A^\dag\A)\L=\cP_\A\widehat{\Y}$ and
\begin{equation*}
\begin{split}
f(\A,\B)&=\frac{1}{2}\|\widehat{\Y}-\A\B\|_F^2\\
&=\frac{1}{2}\|\widehat{\Y}-\cP_\A\widehat{\Y}\|_F^2\\
&=\frac{1}{2}\|\cP_{\A^\perp}\widehat{\Y}\|_F^2\\
&=\frac{1}{2}\sum_{i\not\in\cI}\lambda_i.
\end{split}
\end{equation*}
The claim is proved.
\end{proof}

So the local minimizer of $f(\A,\B)$ is given by \eqref{equ: critical point} with $\cI$ such that $(\lambda_{i_1}, \lambda_{i_2}, \ldots, \lambda_{i_r}) = \Phi$, where $\Phi$ is the sequence of the $r$ largest eigenvalues of $\widehat{\Y}\widehat{\Y}^T$. Such a local minimizer is globally optimal according to \eqref{equ: objective function in spectral}.
We then show that when $\cI$ consists of other combinations of indices of eigenvalues, i.e., $(\lambda_{i_1}, \lambda_{i_2}, \ldots, \lambda_{i_r})\not=\Phi$, the corresponding pair $(\A,\B)$ given by \eqref{equ: critical point} is a strict saddle point.

\begin{claim}
If $\cI$ is such that $(\lambda_{i_1}, \lambda_{i_2}, \ldots, \lambda_{i_r}) \not=\Phi$, then the pair $(\A,\B)$ given by \eqref{equ: critical point} is a strict saddle point.
\end{claim}

\begin{proof}
If $(\lambda_{i_1}, \lambda_{i_2}, \ldots, \lambda_{i_r}) \not=\Phi$, then there exists a $i$ such that $ \lambda_{i_i}$ does not equal the $i$-th element $\lambda_i$ of $\Phi$.
Denote by $\U\D=\mathbf R$. It is enough to slightly perturb the column space of $\A$ towards the direction of an eigenvector of $\lambda_i$. More precisely, let $j$ is the largest index in $\cI$. For any $\epsilon$, let $\widetilde{\mathbf{R}}$ be the matrix such that $\widetilde{\mathbf{R}}_{:k} = \mathbf{R}_{:k}(k\neq j)$ and $\widetilde{\mathbf{R}}_{:j}=(1+\epsilon^2)^{-1/2}(\mathbf{R}_{:j}+\epsilon\mathbf{R}_{:i})$.
Let $\widetilde \A=\widetilde{\mathbf{R}}_\cI\C$ and $\widetilde \B=\widetilde \A^\dag\Y+(\I-\widetilde\A^\dag\widetilde\A)\L$. A direct calculation shows that
\begin{equation*}
f(\widetilde\A,\widetilde\B)=f(\A,\B)-\epsilon^2(\lambda_i-\lambda_j)/(2+2\epsilon^2).
\end{equation*}
Hence,
\begin{equation*}
\lim_{\epsilon\rightarrow 0} \frac{f(\widetilde\A,\widetilde\B)-f(\A,\B)}{\epsilon^2}=-\frac{1}{2}(\lambda_i-\lambda_j)<0,
\end{equation*}
and thus the pair $(\A,\B)$ is a strict saddle point.
\end{proof}
Note that all critical points of $f(\A,\B)$ are in the form of \eqref{equ: critical point}, and if $\cI\not=\Phi$, the pair $(\A,\B)$ given by \eqref{equ: critical point} is a strict saddle point, while if $\cI=\Phi$, then the pair $(\A,\B)$ given by \eqref{equ: critical point} is a local minimum. We conclude that $f(\A,\B)$ has no local maximum. The proof is completed.
\end{proof}

\section{Proof of Lemma \ref{lemma: f and k local minimum}}
\label{section: proof of strong duality}
\noindent{\textbf{Lemma \ref{lemma: f and k local minimum}} (Restated)\textbf{.}}
\emph{Let $(\widetilde\A,\widetilde\B)$ be a global minimizer of $F(\A,\B)$. If there exists a dual certificate $\widetilde\bLambda$ as in Condition \ref{lemma: Lagrangian multiplier} such that the pair $(\widetilde\A,\widetilde\B)$ is a local minimizer of $L(\A,\B,\widetilde\bLambda)$ for the fixed $\widetilde\bLambda$, then strong duality holds. Moreover, we have the relation $\widetilde\A\widetilde\B=\textup{\textsf{svd}}_r(-\widetilde\bLambda)$.}

\begin{proof}
By the assumption of the lemma, $(\widetilde\A,\widetilde\B)$ is a local minimizer of $L(\A,\B,\widetilde\bLambda)=\frac{1}{2}\|-\widetilde\bLambda-\A\B\|_F^2+c(\widetilde\bLambda)$, where $c(\widetilde\bLambda)$ is a function that is independent of $\A$ and $\B$.
So according to Lemma \ref{lemma: local-global}, $(\widetilde\A,\widetilde\B)=\argmin_{\A,\B} L(\A,\B,\widetilde \bLambda)$, namely, $(\widetilde\A,\widetilde\B)$ globally minimizes $L(\A,\B,\bLambda)$ when $\bLambda$ is fixed to $\widetilde\bLambda$. Furthermore, $\widetilde \bLambda\in\partial_{\X} H(\X)|_{\X=\widetilde \A\widetilde \B}$ implies that $\widetilde \A\widetilde \B\in\partial_\bLambda H^*(\bLambda)|_{\bLambda=\widetilde \bLambda}$ by the convexity of function $H$, meaning that $\0\in \partial_\bLambda L(\widetilde \A,\widetilde \B,\bLambda)$. So $\widetilde\bLambda=\argmax_{\bLambda} L(\widetilde\A,\widetilde\B,\bLambda)$ due to the concavity of $L(\widetilde\A,\widetilde\B,\bLambda)$ w.r.t. variable $\bLambda$. Thus $(\widetilde\A,\widetilde\B,\widetilde\bLambda)$ is a primal-dual saddle point of $L(\A,\B,\bLambda)$.

We now prove the strong duality. By the fact that $F(\A,\B)=\max_\bLambda L(\A,\B,\bLambda)$ and that $\widetilde\bLambda=\argmax_{\bLambda} L(\widetilde\A,\widetilde\B,\bLambda)$, we have
\begin{equation*}
F(\widetilde\A,\widetilde\B)=L(\widetilde\A,\widetilde\B,\widetilde\bLambda)\le L(\A,\B,\widetilde\bLambda),\ \ \forall \A,\B.
\end{equation*}
where the inequality holds because $(\widetilde\A,\widetilde\B,\widetilde\bLambda)$ is a primal-dual saddle point of $L$. So on the one hand, we have
\begin{equation*}
\min_{\A,\B}\max_\bLambda L(\A,\B,\bLambda)=F(\widetilde\A,\widetilde\B)\le \min_{\A,\B} L(\A,\B,\widetilde\bLambda)\le\max_\bLambda\min_{\A,\B} L(\A,\B,\bLambda).
\end{equation*}
On the other hand, by weak duality,
\begin{equation*}
\min_{\A,\B}\max_\bLambda L(\A,\B,\bLambda)\ge \max_\bLambda\min_{\A,\B} L(\A,\B,\bLambda).
\end{equation*}
Therefore, $\min_{\A,\B}\max_\bLambda L(\A,\B,\bLambda)=\max_\bLambda\min_{\A,\B} L(\A,\B,\bLambda)$, i.e., strong duality holds. Hence,
\begin{equation*}
\begin{split}
\widetilde\A\widetilde\B&=\argmin_{\A\B} L(\A,\B,\widetilde\bLambda)\\
&=\argmin_{\A\B} \frac{1}{2}\|-\widetilde\bLambda-\A\B\|_F^2-\frac{1}{2}\|\widetilde\bLambda\|_F^2-H^*(\widetilde\bLambda)\\
&=\argmin_{\A\B} \frac{1}{2}\|-\widetilde\bLambda-\A\B\|_F^2\\
&=\textsf{svd}_r(-\widetilde\bLambda),
\end{split}
\end{equation*}
as desired.
\end{proof}

\section{Existence of Dual Certificate for Matrix Completion} \label{section: existence of dual certificate}


Let $\widetilde\A \in \R^{n_1\times r}$ and $\widetilde\B \in \R^{r \times n_2}$ such that $\widetilde\A\widetilde\B=\X^*$.
Then we have the following lemma.

\begin{lemma}
\label{lemma: existence of dual certificate}
Let $\Omega\sim\textup{\textsf{Uniform}}(m)$ be the support set uniformly distributed among all sets of cardinality $m$. Suppose that $m\ge c\kappa^2\mu n_{(1)}r\log n_{(1)}\log_{2\kappa} n_{(1)}$ for an absolute constant $c$ and $\X^*$ obeys $\mu$-incoherence \eqref{equ: incoherence}. Then there exists $\widetilde \bLambda$ such that
\begin{flalign}
\label{equ: dual condition for matrix completion}
\begin{split}
&\mbox{(1)}\qquad \widetilde \bLambda\in\Omega,\\
&\mbox{(2)}\qquad \cP_{\cT}(-\widetilde\bLambda)=\widetilde\A\widetilde\B,\\
&\mbox{(3)}\qquad \|\cP_{\cT^\perp}\widetilde\bLambda\|<\frac{2}{3}\sigma_r(\widetilde\A\widetilde\B).\\
\end{split}
\end{flalign}
with probability at least $1-n_{(1)}^{-10}$.
\end{lemma}

The rest of the section is denoted to the proof of Lemma~\ref{lemma: existence of dual certificate}.
We begin with the following lemma.
\begin{lemma}
If we can construct an $\bLambda$ such that
\begin{flalign}
\label{equ: dual condition for matrix completion 2}
\begin{split}
&\textup{(a)}\qquad \bLambda\in\Omega,\\
&\textup{(b)}\qquad \|\cP_{\cT}(-\bLambda)-\widetilde\A\widetilde\B\|_F\le \sqrt{\frac{r}{3n_{(1)}^2}}\sigma_r(\widetilde\A\widetilde\B),\\
&\textup{(c)}\qquad \|\cP_{\cT^\perp}\bLambda\|<\frac{1}{3}\sigma_r(\widetilde\A\widetilde\B),\\
\end{split}
\end{flalign}
then we can construct an $\widetilde\bLambda$ such that Eqn. \eqref{equ: dual condition for matrix completion} holds with probability at least $1 - n_{(1)}^{-10}$.
\end{lemma}
\begin{proof}
To prove the lemma, we first claim the following theorem.
\begin{theorem}[\cite{Candes2009exact}, Theorem 4.1]
\label{theorem: concentration of operator}
Assume that $\Omega$ is sampled according to the Bernoulli model with success probability $p=\Theta(\frac{m}{n_1n_2})$, and incoherence condition \eqref{equ: incoherence} holds. Then there is an absolute constant $C_R$ such that for $\beta>1$, we have
\begin{equation*}
\|p^{-1}\cP_\cT\cP_\Omega\cP_\cT-\cP_\cT\|\le C_R\sqrt{\frac{\beta\mu n_{(1)}r\log n_{(1)}}{m}}\triangleq \epsilon,
\end{equation*}
with probability at least $1-3n^{-\beta}$ provided that $C_R\sqrt{\frac{\beta\mu n_{(1)}r\log n_{(1)}}{m}} < 1$.
\end{theorem}
Suppose that Condition \eqref{equ: dual condition for matrix completion 2} holds. Let $\Y=\widetilde\bLambda-\bLambda\in\Omega$ be the perturbation matrix between $\bLambda$ and $\widetilde\bLambda$ such that $\cP_\cT(-\widetilde\bLambda)=\widetilde\A\widetilde\B$. Such a $\Y$ exists by setting $\Y=\cP_\Omega\cP_\cT(\cP_\cT\cP_\Omega\cP_\cT)^{-1}(\cP_\cT(-\bLambda)-\widetilde\A\widetilde\B)$. So $\|\cP_{\cT}\Y\|_F\le \sqrt{\frac{r}{3n_{(1)}^2}}\sigma_r(\widetilde\A\widetilde\B)$. We now prove Condition (3) in Eqn. \eqref{equ: dual condition for matrix completion}. Observe that
\begin{equation}
\begin{split}
\|\cP_{\cT^\perp}\widetilde\bLambda\|&\le \|\cP_{\cT^\perp}\bLambda\|+\|\cP_{\cT^\perp}\Y\|\\
&\le \frac{1}{3}\sigma_r(\widetilde\A\widetilde\B)+\|\cP_{\cT^\perp}\Y\|.
\end{split}
\end{equation}
So we only need to show $\|\cP_{\cT^\perp}\Y\|\le\frac{1}{3}\sigma_r(\widetilde\A\widetilde\B)$.

Before proceeding, we begin by introducing a normalized version $\cQ_\Omega:\R^{n_1\times n_2}\rightarrow\R^{n_1\times n_2}$ of $\cP_\Omega$:
\begin{equation*}
\cQ_\Omega=p^{-1}\cP_\Omega-\cI.
\end{equation*}
With this, we have
\begin{equation*}
\cP_\cT\cP_\Omega\cP_\cT=p\cP_\cT(\cI+\cQ_\Omega)\cP_\cT.
\end{equation*}
Note that for any operator $\cP:\cT\rightarrow\cT$, we have
\begin{equation*}
\cP^{-1}=\sum_{k\ge 0} (\cP_\cT-\cP)^k \mbox{ whenever } \|\cP_\cT-\cP\|<1.
\end{equation*}
So according to Theorem \ref{theorem: concentration of operator}, the operator $p(\cP_\cT\cP_\Omega\cP_\cT)^{-1}$ can be represented as a \emph{convergent} Neumann series
\begin{equation*}
p(\cP_\cT\cP_\Omega\cP_\cT)^{-1}=\sum_{k\ge 0}(-1)^k(\cP_\cT\cQ_\Omega\cP_\cT)^k,
\end{equation*}
because $\|\cP_\cT\cQ_\Omega\cP_\cT\|\le \epsilon<\frac{1}{2}$ once $m\ge C\mu n_{(1)}r\log n_{(1)}$ for a sufficiently large absolute constant $C$. We also note that
\begin{equation*}
p(\cP_{\cT^\perp}\cQ_\Omega\cP_\cT)=\cP_{\cT^\perp}\cP_\Omega\cP_\cT,
\end{equation*}
because $\cP_{\cT^\perp}\cP_\cT=0$. Thus
\begin{equation*}
\begin{split}
\|\cP_{\cT^\perp}\Y\|&=\|\cP_{\cT^\perp}\cP_\Omega\cP_\cT(\cP_\cT\cP_\Omega\cP_\cT)^{-1}(\cP_\cT(-\bLambda)-\widetilde\A\widetilde\B))\|\\
&=\|\cP_{\cT^\perp}\cQ_\Omega\cP_\cT p(\cP_\cT\cP_\Omega\cP_\cT)^{-1}((\cP_\cT(-\bLambda)-\widetilde\A\widetilde\B))\|\\
&=\|\sum_{k\ge 0} (-1)^k\cP_{\cT^\perp}\cQ_\Omega(\cP_\cT\cQ_\Omega\cP_\cT)^k((\cP_\cT(-\bLambda)-\widetilde\A\widetilde\B))\|\\
&\le \sum_{k\ge 0}\| (-1)^k\cP_{\cT^\perp}\cQ_\Omega(\cP_\cT\cQ_\Omega\cP_\cT)^k((\cP_\cT(-\bLambda)-\widetilde\A\widetilde\B))\|_F\\
&\le \|\cQ_\Omega\|\sum_{k\ge 0} \|\cP_\cT\cQ_\Omega\cP_\cT\|^k\|\cP_\cT(-\bLambda)-\widetilde\A\widetilde\B))\|_F\\
&\le \frac{4}{p}\|\cP_\cT(-\bLambda)-\widetilde\A\widetilde\B)\|_F\\
&\le \Theta\left(\frac{n_1n_2}{m} \right) \sqrt{\frac{r}{3 n_{(1)}^2}}\sigma_r(\widetilde\A\widetilde\B)\\
&\le \frac{1}{3}\sigma_r(\widetilde\A\widetilde\B)
\end{split}
\end{equation*}
with high probability. The proof is completed.
\end{proof}

It thus suffices to construct a dual certificate $\bLambda$ such that all conditions in \eqref{equ: dual condition for matrix completion 2} hold. To this end, partition $\Omega=\Omega_1\cup\Omega_2\cup...\cup\Omega_b$ into $b$ partitions of size $q$. By assumption, we may choose
\begin{equation*}
q\ge \frac{128}{3} C \beta\kappa^2\mu rn_{(1)}\log n_{(1)}\quad\mbox{and}\quad b\ge\frac{1}{2}\log_{2\kappa} \left(24^2 n_{(1)}^2\kappa^2\right)
\end{equation*}
for a sufficiently large constant $C$.
Let $\Omega_j\sim\mathsf{Ber}(q)$ denote the set of indices corresponding to the $j$-th partitions.
Define $\W_0=\widetilde\A\widetilde\B$ and set $\bLambda_k=\frac{n_1n_2}{q}\sum_{j=1}^k\cP_{\Omega_j}(\W_{j-1})$, $\W_k=\widetilde\A\widetilde\B-\cP_{\cT}(\bLambda_k)$ for $k=1,2,...,b$. Then by Theorem \ref{theorem: concentration of operator},
\begin{equation*}
\begin{split}
\|\W_k\|_F&=\left\|\W_{k-1}-\frac{n_1n_2}{q}\cP_\cT\cP_{\Omega_k}(\W_{k-1})\right\|_F=\left\|\left(\cP_\cT-\frac{n_1n_2}{q}\cP_\cT\cP_{\Omega_k}\cP_\cT\right)(\W_{k-1})\right\|_F\\&\le\frac{1}{2\kappa}\|\W_{k-1}\|_F.
\end{split}
\end{equation*}
So it follows that $\|\widetilde\A\widetilde\B-\cP_{\cT}(\bLambda_b)\|_F=\|\W_b\|_F\le (2\kappa)^{-b}\|\W_0\|_F\le (2\kappa)^{-b}\sqrt{r}\sigma_1(\widetilde\A\widetilde\B)\le \sqrt{\frac{r}{24^2  n_{(1)}^2}}\sigma_r(\widetilde\A\widetilde\B)$.

The following lemma together implies the strong duality of \eqref{equ: equality constraint matrix completion} straightforwardly.
\begin{lemma}
\label{lemma: dual condition (c) for matrix completion}
Under the assumptions of Theorem \ref{theorem: matrix completion}, the dual certification $\bLambda_b$ obeys the dual condition \eqref{equ: dual condition for matrix completion 2} with probability at least $1-n_{(1)}^{-10}$.
\end{lemma}

\begin{proof}
It is well known that for matrix completion, the Uniform model $\Omega \sim \textsf{Uniform}(m)$ is equivalent to the Bernoulli model $\Omega\sim\textup{\textsf{Ber}}(p)$, where each element in $[n_1] \times [n_2]$ is included with probability $p = \Theta(m/(n_1 n_2))$ independently; see Section~\ref{section: Equivalence of Bernoulli and Uniform Models} for a brief justification.
By the equivalence, we can suppose $\Omega\sim\textup{\textsf{Ber}}(p)$.

To prove Lemma \ref{lemma: dual condition (c) for matrix completion}, as a preliminary, we need the following lemmas.

\begin{lemma}[\cite{chen2015incoherence}, Lemma 2]
\label{lemma: 2 norm and infty norm}
Suppose $\Z$ is a fixed matrix. Suppose $\Omega\sim\textup{\textsf{Ber}}(p)$. Then with high probability,
\begin{equation*}
\|(\cI-p^{-1}\cP_\Omega)\Z\|\le C_0'\left(\frac{\log n_{(1)}}{p}\|\Z\|_\infty+\sqrt{\frac{\log n_{(1)}}{p}}\|\Z\|_{\infty,2}\right),
\end{equation*}
where $C_0'>0$ is an absolute constant and
\begin{equation*}
\|\Z\|_{\infty,2}=\max\left\{\max_i\sqrt{\sum_b \Z_{ib}^2},\max_j\sqrt{\sum_a\Z_{aj}^2}\right\}.
\end{equation*}
\end{lemma}

\begin{lemma}[\cite{Candes}, Lemma 3.1]
\label{lemma: infty norm and infty norm}
Suppose $\Omega\sim\textup{\textsf{Ber}}(p)$ and $\Z$ is a fixed matrix. Then with high probability,
\begin{equation*}
\|\Z-p^{-1}\cP_\cT\cP_\Omega\Z\|_\infty\le \epsilon\|\Z\|_\infty,
\end{equation*}
provided that $p\ge C_0\epsilon^{-2}(\mu r\log n_{(1)})/n_{(2)}$ for some absolute constant $C_0>0$.
\end{lemma}

\begin{lemma}[\cite{chen2015incoherence}, Lemma 3]
\label{lemma: infty 2 norm and infty, infty 2 norm}
Suppose that $\Z$ is a fixed matrix and $\Omega\sim\textup{\textsf{Ber}}(p)$. If $p\ge c_0 \mu r\log n_{(1)}/n_{(2)}$ for some $c_0$ sufficiently large, then with high probability,
\begin{equation*}
\|(p^{-1}\cP_{\cT}\cP_\Omega-\cP_{\cT})\Z\|_{\infty,2}\le \frac{1}{2}\sqrt{\frac{n_{(1)}}{\mu r}}\|\Z\|_\infty+\frac{1}{2}\|\Z\|_{\infty,2}.
\end{equation*}
\end{lemma}

Observe that by Lemma \ref{lemma: infty norm and infty norm},
\begin{equation*}
\|\W_j\|_\infty\le \left(\frac{1}{2}\right)^j\|\widetilde\A\widetilde\B\|_\infty,
\end{equation*}
and by Lemma \ref{lemma: infty 2 norm and infty, infty 2 norm},
\begin{equation*}
\|\W_j\|_{\infty,2}\le \frac{1}{2}\sqrt{\frac{n_{(1)}}{\mu r}}\|\W_{j-1}\|_\infty+\frac{1}{2}\|\W_{j-1}\|_{\infty,2}.
\end{equation*}
So
\begin{equation*}
\begin{split}
&\ \ \ \|\W_j\|_{\infty,2}\\&\le \left(\frac{1}{2}\right)^j\sqrt{\frac{n_{(1)}}{\mu r}}\|\widetilde\A\widetilde\B\|_\infty+\frac{1}{2}\|\W_{j-1}\|_{\infty,2}\\
&\le j\left(\frac{1}{2}\right)^j\sqrt{\frac{n_{(1)}}{\mu r}}\|\widetilde\A\widetilde\B\|_\infty+\left(\frac{1}{2}\right)^j\|\widetilde\A\widetilde\B\|_{\infty,2}.
\end{split}
\end{equation*}
Therefore,
\begin{equation*}
\begin{split}
&\ \ \ \ \|\cP_{\cT^\perp}\bLambda_b\|\\
&\le \sum_{j=1}^b\|\frac{n_1n_2}{q}\cP_{\cT^\perp}\cP_{\Omega_j}\W_{j-1}\|\\
&=\sum_{j=1}^b\|\cP_{\cT^\perp}(\frac{n_1n_2}{q}\cP_{\Omega_j}\W_{j-1}-\W_{j-1})\|\\
&\le\sum_{j=1}^b\|(\frac{n_1n_2}{q}\cP_{\Omega_j}-\cI)(\W_{j-1})\|.
\end{split}
\end{equation*}
Let $p$ denote $\Theta\left(\frac{q}{n_1n_2}\right)$. By Lemma~\ref{lemma: 2 norm and infty norm},
\begin{equation*}
\begin{split}
&\ \ \ \ \|\cP_{\cT^\perp}\bLambda_b\|\\
&\le C_0'\frac{\log n_{(1)}}{p}\sum_{j=1}^b \|\W_{j-1}\|_\infty+C_0'\sqrt{\frac{\log n_{(1)}}{p}}\sum_{j=1}^b \|\W_{j-1}\|_{\infty,2}\\
&\le C_0'\frac{\log n_{(1)}}{p}\sum_{j=1}^b \left(\frac{1}{2}\right)^j\|\widetilde\A\widetilde\B\|_\infty+C_0'\sqrt{\frac{\log n_{(1)}}{p}}\sum_{j=1}^b \left[j\left(\frac{1}{2}\right)^j\sqrt{\frac{n_{(1)}}{\mu r}}\|\widetilde\A\widetilde\B\|_\infty+\left(\frac{1}{2}\right)^j\|\widetilde\A\widetilde\B\|_{\infty,2}\right]\\
&\le C_0'\frac{\log n_{(1)}}{p}\|\widetilde\A\widetilde\B\|_\infty+2C_0'\sqrt{\frac{\log n_{(1)}}{p}}\sqrt{\frac{n_{(1)}}{\mu r}}\|\widetilde\A\widetilde\B\|_\infty+C_0'\sqrt{\frac{\log n_{(1)}}{p}}\|\widetilde\A\widetilde\B\|_{\infty,2}.
\end{split}
\end{equation*}

Setting $\widetilde\A\widetilde\B=\X^*$, we note the facts that (we assume WLOG $n_2\ge n_1$)
\begin{equation*}
\|\X^*\|_{\infty,2}=\max_i \|\e_i^T\U\mathbf{\Sigma}\V^T\|_2\le \max_i\|\e_i^T\U\|\sigma_1(\X^*)\le \sqrt{\frac{\mu r}{n_1}}\sigma_1(\X^*) \le \sqrt{\frac{\mu r}{n_1}} \kappa \sigma_r(\X^*),
\end{equation*}
and that
\begin{equation*}
\begin{split}
\|\X^*\|_\infty&=\max_{ij}\langle\X^*,\e_i\e_j^T\rangle=\max_{ij}\langle\U\mathbf{\Sigma}\V^T,\e_i\e_j^T\rangle=\max_{ij}\langle\e_i^T\U\mathbf{\Sigma},\e_j^T\V\rangle\\
&\le \max_{ij}\|\e_i^T\U\mathbf{\Sigma}\V^T\|_2\|\e_j^T\V\|_2\le \max_{j} \|\X^*\|_{\infty,2}\|\e_j^T\V\|_2\le \frac{\mu r\kappa}{\sqrt{n_1n_2}}\sigma_r(\X^*).
\end{split}
\end{equation*}
Substituting $p=\Theta\left(\frac{\kappa^2\mu r n_{(1)}\log (n_{(1)})\log_{2\kappa}(n_{(1)})}{n_1n_2}\right)$, we obtain $\|\cP_{\cT^\perp}\bLambda_b\|<\frac{1}{3}\sigma_r(\X^*)$. The proof is completed.
\end{proof}

\section{Subgradient of the $r*$ Function}

\begin{lemma}
\label{lemma: subgradient of r*}
Let $\U\mathbf{\Sigma}\V^T$ be the skinny SVD of matrix $\X^*$ of rank $r$. The subdifferential of $\|\cdot\|_{r*}$ evaluated at $\X^*$ is given by
\begin{equation*}
\partial \|\X^*\|_{r*}=\{\X^*+\W: \U^T\W=\0,\W\V=\0,\|\W\|\le \sigma_r(\X^*)\}.
\end{equation*}
\end{lemma}

\begin{proof}
Note that for any fixed function $f(\cdot)$, the set of all optimal solutions of the problem
\begin{equation}
\label{equ: definition of conjugate function}
f^*(\X^*)=\max_\Y \langle\X^*,\Y\rangle-f(\Y)
\end{equation}
form the subdifferential of the conjugate function $f^*(\cdot)$ evaluated at $\X^*$. Set $f(\cdot)$ to be $\frac{1}{2}\|\cdot\|_r^2$ and notice that the function $\frac{1}{2}\|\cdot\|_r^2$ is unitarily invariant. By Von Neumann's trace inequality, the optimal solutions to problem \eqref{equ: definition of conjugate function} are given by $[\U,\U^\perp]\diag([\sigma_1(\Y),...,\sigma_r(\Y),\sigma_{r+1}(\Y),...,\sigma_{n_{(2)}}(\Y)])[\V,\V^\perp]^T$, where $\{\sigma_i(\Y)\}_{i=r+1}^{n_{(2)}}$ can be any value no larger than $\sigma_r(\Y)$ and $\{\sigma_i(\Y)\}_{i=1}^{r}$ are given by the optimal solution to the problem
\begin{equation*}
\max_{\{\sigma_i(\Y)\}_{i=1}^r} \sum_{i=1}^r\sigma_i(\X^*)\sigma_i(\Y)-\frac{1}{2}\sum_{i=1}^r\sigma_i^2(\Y).
\end{equation*}
The solution is unique such that $\sigma_i(\Y)=\sigma_i(\X^*)$, $i=1,2,...,r$. The proof is complete.
\end{proof}

\section{Proof of Theorem \ref{theorem: uniqueness of matrix completion}}
\label{section: proof of information-theoretic upper bound for matrix completion}

\medskip
\noindent{\textbf{Theorem~\ref{theorem: uniqueness of matrix completion}} (Uniqueness of Solution. Restated)\textbf{.}}
\emph{
Let $\Omega\sim\textup{\textsf{Uniform}}(m)$ be the support set uniformly distributed among all sets of cardinality $m$. Suppose that $m\ge c\kappa^2\mu n_{(1)}r\log n_{(1)}\log_{2\kappa} n_{(1)}$ for an absolute constant $c$ and $\X^*$ obeys $\mu$-incoherence \eqref{equ: incoherence}. Then $\X^*$ is the unique solution of non-convex optimization
\begin{equation}
\label{equ: non-convex mc}
\min_{\A,\B} \frac{1}{2}\|\A\B\|_F^2,\quad \mathrm{s.t.}\quad \cP_\Omega(\A\B)=\cP_\Omega(\X^*),
\end{equation}
with probability at least $1-n_{(1)}^{-10}$.
}

\begin{proof}
We note that a recovery result under the Bernoulli model automatically implies a
corresponding result for the uniform model~\cite{Candes}; see Section \ref{section: Equivalence of Bernoulli and Uniform Models} for the details. So in the following, we assume the Bernoulli model.

Consider the feasibility of the matrix completion problem:
\begin{equation}
\label{equ: rank-constraint matrix completion}
\mbox{Find a matrix $\X\in\R^{n_1\times n_2}$ such that}\quad\cP_\Omega(\X)=\cP_\Omega(\X^*),\quad \|\X\|_F\le \|\X^*\|_F, \quad\rank(\X)\le r.
\end{equation}
Note that if $\X^*$ is the unique solution of \eqref{equ: rank-constraint matrix completion}, then $\X^*$ is the unique solution of \eqref{equ: non-convex mc}. We now show the former.
Our proof first identifies a feasibility condition for problem \eqref{equ: rank-constraint matrix completion}, and then shows that $\X^*$ is the only matrix that obeys this feasibility condition when the sample size is large enough. We denote by
\begin{equation*}
\cD_\cS(\X^*)=\{\X-\X^*\in\R^{n_1\times n_2}:\rank(\X)\le r,\ \|\X\|_F\le \|\X^*\|_F\},
\end{equation*}
and
\begin{equation*}
\cT= \{\U\X^T+\Y\V^T,\ \X\in\R^{n_2\times r},\ \Y\in\R^{n_1\times r}\},
\end{equation*}
where $\U\mathbf{\Sigma}\V^T$ is the skinny SVD of $\X^*$.

We have the following proposition for the feasibility of problem \eqref{equ: rank-constraint matrix completion}.
\begin{proposition}[Feasibility Condition]
\label{proposition: feasibility condition for matrix completion}
$\X^*$ is the unique feasible solution to problem \eqref{equ: rank-constraint matrix completion} if $\cD_\cS(\X^*)\cap\Omega^\perp=\{\0\}$.
\end{proposition}
\begin{proof}
Notice that problem \eqref{equ: rank-constraint matrix completion} is equivalent to another feasibility problem
\begin{equation*}
\mbox{Find a matrix $\D\in\R^{n_1\times n_2}$ such that}\quad \rank(\X^*+\D)\le r,\ \ \|\X^*+\D\|_F\le \|\X^*\|_F,\ \ \D\in \Omega^\perp.
\end{equation*}
Suppose that $\cD_\cS(\X^*)\cap\Omega^\perp=\{\0\}$. Since $\rank(\X^*+\D)\le r$ and $\|\X^*+\D\|_F\le \|\X^*\|_F$ are equivalent to $\D\in \cD_\cS(\X^*)$, and note that $\D\in\Omega^\perp$, we have $\D=\0$, which means $\X^*$ is the unique feasible solution to problem \eqref{equ: rank-constraint matrix completion}.
\end{proof}

The remainder of the proof is to show $\cD_\cS(\X^*)\cap\Omega^\perp=\{\0\}$. To proceed, we note that
\begin{equation*}
\begin{split}
\cD_\cS(\X^*)&=\left\{\X - \X^* \in\R^{n_1\times n_2}:\rank(\X)\le r,\ \frac{1}{2}\|\X\|_F^2\le \frac{1}{2}\|\X^*\|_F^2\right\}\\
&\subseteq \{\X- \X^*\in\R^{n_1\times n_2}:\|\X\|_{r*}\le \|\X^*\|_{r*}\} \ \ \left(\text{since }\frac{1}{2}\|\Y\|_F^2=\|\Y\|_{r*}\text{ for any rank-$r$ matrix}\right)\\
&\triangleq \cD_{\cS_*}(\X^*).
\end{split}
\end{equation*}
We now show that
\begin{equation}
\label{equ: convex intersection}
\cD_{\cS_*}(\X^*)\cap \Omega^\perp=\{\0\},
\end{equation}
when $m\ge c\kappa^2\mu rn_{(1)}\log_{2\kappa} (n_{(1)})\log (n_{(1)})$, which will prove $\cD_{\cS}(\X^*)\cap \Omega^\perp=\{\0\}$ as desired.

By Lemma~\ref{lemma: existence of dual certificate}, there exists a $\bLambda$ such that
\begin{flalign*}
\label{equ: non-convex dual condition for matrix completion}
\begin{split}
&\mbox{(1)}\qquad \bLambda\in\Omega,\\
&\mbox{(2)}\qquad \cP_{\cT}(-\bLambda)=\X^*,\\
&\mbox{(3)}\qquad \|\cP_{\cT^\perp}\bLambda\|<\frac{2}{3}\sigma_r(\X^*).
\end{split}
\end{flalign*}
Consider any $\D\in\Omega^\perp$ such that $\D\neq \0$.
By Lemma \ref{lemma: subgradient of r*}, for any $\W\in\cT^\perp \text{ and }\|\W\|\le \sigma_r(\X^*)$,
\begin{equation*}
\begin{split}
\|\X^*+\D\|_{r*}&\ge \|\X^*\|_{r*}+\langle\X^*+\W,\D\rangle.
\end{split}
\end{equation*}
Since $\langle\W,\D\rangle = \langle\cP_{\cT^\perp}\W,\D\rangle = \langle\W, \cP_{\cT^\perp} \D\rangle$, we can choose $ \W$ such that $ \langle\W,\D\rangle=\sigma_r(\X^*)\|\cP_{\cT^\perp}\D\|_*$. Then
\begin{equation*}
\begin{split}
\|\X^*+\D\|_{r*}
&\ge\|\X^*\|_{r*}+\sigma_r(\X^*)\|\mathcal{P}_{\cT^\perp} \D\|_*+\langle\X^*,\D\rangle \\
&=\|\X^*\|_{r*}+\sigma_r(\X^*)\|\mathcal{P}_{\cT^\perp} \D\|_*+\langle\X^*+\bLambda,\D\rangle\quad(\text{since }\bLambda\in\Omega\text{ and }\D\in\Omega^\perp)\\
&=\|\X^*\|_{r*}+\sigma_r(\X^*)\|\mathcal{P}_{\cT^\perp} \D\|_*+\langle\X^*+\cP_\cT\bLambda,\D\rangle+\langle\cP_{\cT^\perp}\bLambda,\D\rangle\\
&=\|\X^*\|_{r*}+\sigma_r(\X^*)\|\mathcal{P}_{\cT^\perp} \D\|_*+\langle\cP_{\cT^\perp}\bLambda,\D\rangle\quad(\text{by condition (2)})\\
&=\|\X^*\|_{r*}+\sigma_r(\X^*)\|\mathcal{P}_{\cT^\perp} \D\|_*+\langle\cP_{\cT^\perp}\cP_{\cT^\perp}\bLambda,\D\rangle \\
&=\|\X^*\|_{r*}+\sigma_r(\X^*)\|\mathcal{P}_{\cT^\perp} \D\|_*+\langle\cP_{\cT^\perp}\bLambda,\cP_{\cT^\perp}\D\rangle \\
&\ge \|\X^*\|_{r*}+\sigma_r(\X^*)\|\mathcal{P}_{\cT^\perp} \D\|_*-\|\cP_{\cT^\perp}\bLambda\|\|\cP_{\cT^\perp}\D\|_*\quad(\text{by H$\mathrm{\ddot{o}}$lder's inequality})\\
&\ge \|\X^*\|_{r*}+\frac{1}{3}\sigma_r(\X^*)\|\cP_{\cT^\perp}\D\|_*\quad(\text{by condition (3)}).
\end{split}
\end{equation*}
So if $\cT\cap\Omega^\perp=\{\0\}$, since $\D\in\Omega^\perp$ and $\D\not=\0$, we have $\D\not\in\cT$. Therefore,
\begin{equation*}
\|\X^*+\D\|_{r*}>\|\X^*\|_{r*}
\end{equation*}
which then leads to $\cD_{\cS_*}(\X^*)\cap \Omega^\perp=\{\0\}$.

The rest of proof is to show that $\cT\cap\Omega^\perp=\{\0\}$. We have the following lemma.

\begin{lemma}
\label{lemma: operator norm}
Assume that $\Omega\sim \textup{\textsf{Ber}}(p)$ and the incoherence condition \eqref{equ: incoherence} holds. Then with probability at least $1-n_{(1)}^{-10}$, we have $\|\cP_{\Omega^\perp}\cP_{\cT}\|\le \sqrt{1-p+\epsilon p}$, provided that $p\ge C_0\epsilon^{-2}(\mu r\log n_{(1)})/n_{(2)}$, where $C_0$ is an absolute constant.
\end{lemma}
\begin{proof}
If $\Omega\sim\mathsf{Ber}(p)$, we have, by Theorem \ref{theorem: concentration of operator}, that with high probability
\begin{equation*}
\|\cP_\cT-p^{-1}\cP_\cT\cP_\Omega\cP_\cT\|\le\epsilon,
\end{equation*}
provided that $p\ge C_0\epsilon^{-2}\frac{\mu r\log n_{(1)}}{n_{(2)}}$. Note, however, that since $\cI=\cP_\Omega+\cP_{\Omega^\perp}$,
\begin{equation*}
\cP_{\cT}-p^{-1}\cP_\cT\cP_\Omega\cP_\cT=p^{-1}(\cP_\cT\cP_{\Omega^\perp}\cP_\cT-(1-p)\cP_\cT)
\end{equation*}
and, therefore, by the triangle inequality
\begin{equation*}
\|\cP_\cT\cP_{\Omega^\perp}\cP_\cT\|\le \epsilon p+(1-p).
\end{equation*}
Since $\|\cP_{\Omega^\perp}\cP_\cT\|^2 \le \|\cP_\cT\cP_{\Omega^\perp}\cP_\cT\|$, the proof is completed.
\end{proof}

We note that $\|\cP_{\Omega^\perp}\cP_{\cT}\|<1$ implies $\Omega^\perp\cap\cT=\{\0\}$. The proof is completed.
\end{proof}

\section{Proof of Theorem~\ref{theorem: matrix completion}}
\label{section: Dual Certification by Least Squares}

We have shown in Theorem \ref{theorem: uniqueness of matrix completion} that the problem
$
(\widetilde\A,\widetilde\B)=\argmin_{\A,\B} \frac{1}{2}\|\A\B\|_F^2,\textup{\mbox{s.t.}}\cP_\Omega(\A\B)=\cP_\Omega(\X^*),
$
exactly recovers $\X^*$, i.e., $\widetilde\A\widetilde\B=\X^*$, with nearly optimal sample complexity. So if strong duality holds, this non-convex optimization problem can be equivalently converted to the convex program \eqref{equ: equality constraint matrix completion}. Then Theorem \ref{theorem: matrix completion} is straightforward from strong duality.

It now suffices to apply our unified framework in Section \ref{section: Framework} to prove the strong duality.
Let
$$
H(\X)=\I_{\{\M\in\R^{n_1\times n_2}:\ \cP_\Omega\M=\cP_\Omega\X^*\}}(\X)
$$
in Problem \textup{(\textbf{P})}, and let $(\widetilde\A,\widetilde\B)$ be a global solution to the problem.
Then by Theorem \ref{theorem: uniqueness of matrix completion}, $\widetilde\A\widetilde\B=\X^*$.
For Problem \textup{(\textbf{P})} with this special $H(\X)$, we have
\begin{equation*}
\begin{split}
\Psi &= \partial H(\widetilde\A\widetilde\B)
=\{\G\in\R^{n_1\times n_2}: \langle\G,\widetilde\A\widetilde\B\rangle\ge\langle \G,\Y\rangle,\ \mbox{for any } \Y\in\R^{n_1\times n_2}\mbox{ s.t. } \cP_\Omega\Y=\cP_\Omega\X^*\}\\
&=\{\G\in\R^{n_1\times n_2}: \langle\G,\X^*\rangle\ge\langle \G,\Y\rangle,\ \mbox{for any } \Y\in\R^{n_1\times n_2}\mbox{ s.t. } \cP_\Omega\Y=\cP_\Omega\X^*\}
=\Omega,
\end{split}
\end{equation*}
where the third equality holds since $\widetilde\A\widetilde\B=\X^*$.
Combining with Lemma~\ref{lemma: existence of dual certificate} shows that the dual condition in Theorem \ref{theorem: strong duality with condition (c)} holds with high probability, which leads to strong duality and thus proving Theorem \ref{theorem: matrix completion}.

\section{Proof of Theorem \ref{theorem: robust PCA}}

\noindent{\textbf{Theorem \ref{theorem: robust PCA}} (Robust PCA. Restated)\textbf{.}}
\emph{
Suppose $\X^*$ is an $n_1\times n_2$ matrix of rank $r$, and obeys incoherence \eqref{equ: incoherence} and \eqref{equ: strong incoherence for RPCA}.
Assume that the support set $\Omega$ of $\S^*$ is uniformly distributed among all sets of cardinality $m$. Then with probability at least $1-cn_{(1)}^{-10}$, the output of the optimization problem
\begin{equation}
\label{equ: RPCA}
(\widetilde\X,\widetilde\S)=\argmin_{\X,\S} \|\X\|_{r*}+\lambda \|\S\|_1,\quad\textup{\mbox{s.t.}}\quad \D=\X+\S,
\end{equation}
with $\lambda=\frac{\sigma_r(\X^*)}{\sqrt{n_{(1)}}}$ is exact, namely, $\widetilde\X=\X^*$ and $\widetilde\S=\S^*$, provided that
$\rank(\X^*)\le \rho_r\frac{n_{(2)}}{\mu \log^2 n_{(1)}}\ \mbox{and}\ m\le\rho_sn_1n_2$,
where $c$, $\rho_r$, and $\rho_s$ are all positive absolute constants, and function $\|\cdot\|_{r*}$ is given by \eqref{equ: r* norm}.
}

\subsection{Dual Certificates}
\begin{lemma}
\label{lemma: dual conditions for RPCA}
Assume that $\|\cP_\Omega\cP_\cT\|\le 1/2$ and $\lambda<\sigma_r(\X^*)$. Then $(\X^*,\S^*)$ is the unique solution to problem \eqref{theorem: robust PCA} if there exists $(\W,\F,\mathbf{K})$ for which
\begin{equation*}
\X^*+\W=\lambda(\sign(\S^*)+\F+\cP_\Omega\mathbf{K}),
\end{equation*}
where $\W\in\cT^\perp$, $\|\W\|\le \frac{\sigma_r(\X^*)}{2}$, $\F\in\Omega^\perp$, $\|\F\|_\infty\le\frac{1}{2}$, and $\|\cP_\Omega\mathbf{K}\|_F\le\frac{1}{4}$.
\end{lemma}

\begin{proof}
Let $(\X^*+\H,\S^*-\H)$ be any optimal solution to problem \eqref{equ: RPCA}. Denote by $\X^*+\W^*$ an arbitrary subgradient of the $r*$ function at $\X^*$ (see Lemma \ref{lemma: subgradient of r*}), and $\sign(\S^*)+\F^*$ an arbitrary subgradient of the $\ell_1$ norm at $\S^*$. By the definition of the subgradient, the inequality follows
\begin{equation*}
\begin{split}
\|\X^*+\H\|_{r*}+\lambda&\|\S^*-\H\|_1\ge \|\X^*\|_{r*}+\lambda\|\S^*\|_1+\langle\X^*+\W^*,\H\rangle-\lambda\langle\sign(\S^*)+\F^*,\H\rangle\\
&=\|\X^*\|_{r*}+\lambda\|\S^*\|_1+\langle\X^*-\lambda\sign(\S^*),\H\rangle+\langle\W^*,\H\rangle-\lambda\langle\F^*,\H\rangle\\
&=\|\X^*\|_{r*}+\lambda\|\S^*\|_1+\langle\X^*-\lambda\sign(\S^*),\H\rangle+\sigma_r(\X^*)\|\cP_{\cT^\perp}\H\|_*+\lambda\|\cP_{\Omega^\perp}\H\|_1 \\
&=\|\X^*\|_{r*}+\lambda\|\S^*\|_1+\langle\lambda\F+\lambda\cP_\Omega\mathbf{K}-\W,\H\rangle+\sigma_r(\X^*)\|\cP_{\cT^\perp}\H\|_*+\lambda\|\cP_{\Omega^\perp}\H\|_1\\
&\ge\|\X^*\|_{r*}+\lambda\|\S^*\|_1+\frac{\sigma_r(\X^*)}{2}\|\cP_{\cT^\perp}\H\|_*+\frac{\lambda}{2}\|\cP_{\Omega^\perp}\H\|_1-\frac{\lambda}{4}\|\cP_\Omega \H\|_F,
\end{split}
\end{equation*}
where the third line holds by picking $\W^*$ such that $\langle\W^*,\H\rangle=\sigma_r(\X^*)\|\cP_{\cT^\perp}\H\|_*$ and $\langle\F^*,\H\rangle=-\|\cP_{\Omega^\perp}\H\|_1$.\footnote{For instance, $\F^*=-\sign(\cP_{\Omega^\perp}\H)$ is such as matrix. Also, by the duality between the nuclear norm and the operator norm, there is a matrix obeying $\|\W\|=\sigma_r(\X^*)$ such that $\langle\W,\cP_{\cT^\perp}\H\rangle=\sigma_r(\X^*)\|\cP_{\cT^\perp}\H\|_*$. We pick $\W^*=\cP_{\cT^\perp}\W$ here.}
We note that
\begin{equation*}
\begin{split}
\|\cP_\Omega\H\|_F&\le \|\cP_\Omega\cP_\cT\H\|_F+\|\cP_\Omega\cP_{\cT^\perp}\H\|_F\\
&\le \frac{1}{2}\|\H\|_F+\|\cP_{\cT^\perp}\H\|_F\\
&\le \frac{1}{2}\|\cP_\Omega \H\|_F+\frac{1}{2}\|\cP_{\Omega^\perp} \H\|_F+\|\cP_{\cT^\perp}\H\|_F,
\end{split}
\end{equation*}
which implies that $\frac{\lambda}{4}\|\cP_\Omega\H\|_F\le \frac{\lambda}{4}\|\cP_{\Omega^\perp}\H\|_F+\frac{\lambda}{2}\|\cP_{\cT^\perp}\H\|_F\le \frac{\lambda}{4}\|\cP_{\Omega^\perp}\H\|_1+\frac{\lambda}{2}\|\cP_{\cT^\perp}\H\|_*$. Therefore,
\begin{equation*}
\begin{split}
\|\X^*+\H\|_{r*}+\lambda&\|\S^*-\H\|_1\ge \|\X^*\|_{r*}+\lambda\|\S^*\|_1+\frac{\sigma_r(\X^*)-\lambda}{2}\|\cP_{\cT^\perp}\H\|_*+\frac{\lambda}{4}\|\cP_{\Omega^\perp}\H\|_1\\
&\ge \|\X^*+\H\|_{r*}+\lambda\|\S^*-\H\|_1+\frac{\sigma_r(\X^*)-\lambda}{2}\|\cP_{\cT^\perp}\H\|_*+\frac{\lambda}{4}\|\cP_{\Omega^\perp}\H\|_1,
\end{split}
\end{equation*}
where the second inequality holds because $(\X^*+\H,\S^*-\H)$ is optimal. Thus $\H\in\cT\cap\Omega$. Note that $\|\cP_\Omega\cP_\cT\|<1$ implies $\cT\cap\Omega=\{0\}$ and thus $\H=0$. This completes the proof.
\end{proof}

According to Lemma \ref{lemma: dual conditions for RPCA}, to show the exact recoverability of problem \eqref{equ: RPCA}, it is sufficient to find an appropriate $\W$ for which
\begin{equation}
\label{equ: dual conditions for W}
\begin{cases}
\W\in \cT^\perp,\\
\|\W\|\le \frac{\sigma_r(\X^*)}{2},\\
\|\cP_\Omega(\X^*+\W-\lambda\sign(S^*))\|_F\le \frac{\lambda}{4},\\
\|\cP_{\Omega^\perp}(\X^*+\W)\|_\infty\le \frac{\lambda}{2},
\end{cases}
\end{equation}
under the assumptions that $\|\cP_\Omega\cP_\cT\|\le 1/2$ and $\lambda<\sigma_r(\X^*)$. We note that $\lambda=\frac{\sigma_r(\X^*)}{\sqrt{n_{(1)}}}<\sigma_r(\X^*)$. To see $\|\cP_\Omega\cP_\cT\|\le 1/2$, we have the following lemma.
\begin{lemma}[\cite{Candes}, Cor 2.7]
\label{lemma: small operator norm of P_Omega*P_T}
Suppose that $\Omega\sim\ber(p)$ and incoherence \eqref{equ: incoherence} holds. Then with probability at least $1-n_{(1)}^{-10}$, $\|\cP_{\Omega}\cP_\cT\|^2\le p+\epsilon$, provided that $1-p\ge C_0\epsilon^{-2}\mu r\log n_{(1)}/n_{(2)}$ for an absolute constant $C_0$.
\end{lemma}

Setting $p$ and $\epsilon$ as small constants in Lemma \ref{lemma: small operator norm of P_Omega*P_T}, we have $\|\cP_\Omega\cP_\cT\|\le 1/2$ with high probability.

\subsection{Dual Certification by Least Squares and the Golfing Scheme}
The remainder of the proof is to construct $\W$ such that the dual condition \eqref{equ: dual conditions for W} holds true. Before introducing our construction, we assume $\Omega\sim\ber(p)$, or equivalently $\Omega^\perp\sim\ber(1-p)$, where $p$ is allowed be as large as an absolute constant. Note that $\Omega^\perp$ has the same distribution as that of $\Omega_1\cup\Omega_2\cup...\cup\Omega_{j_0}$, where the $\Omega_j$'s are drawn independently with replacement from $\ber(q)$, $j_0=\lceil\log n_{(1)}\rceil$, and $q$ obeys $p=(1-q)^{j_0}$ ($q=\Omega(1/\log n_{(1)})$ implies $p=\cO(1)$). We construct $\W$ based on such a distribution.

Our construction separates $\W$ into two terms: $\W=\W^L+\W^S$. To construct $\W^L$, we apply the golfing scheme introduced by \cite{Gross2011recovering,recht2011simpler}. Specifically, $\W^L$ is constructed by an inductive procedure:
\begin{equation}
\label{equ: golfing scheme for RPCA}
\begin{split}
\Y_j=\Y_{j-1}&+q^{-1}\cP_{\Omega_j}\cP_\cT(\X^*-\Y_{j-1}),\ \Y_0=\0,\\
&\W^L=\cP_{\cT^\perp}\Y_{j_0}.
\end{split}
\end{equation}
To construct $\W^S$, we apply the method of least squares by \cite{Candes}, which is
\begin{equation}
\label{equ: least squares for RPCA}
\W^S=\lambda\cP_{\cT^\perp}\sum_{k\ge 0}(\cP_\Omega\cP_\cT\cP_\Omega)^k\sign(\S^*).
\end{equation}
Note that $\|\cP_\Omega\cP_\cT\|\le 1/2$. Thus $\|\cP_\Omega\cP_\cT\cP_\Omega\|\le 1/4$ and the Neumann series in \eqref{equ: least squares for RPCA} is well-defined. Observe that $\cP_\Omega \W^S=\lambda(\cP_\Omega-\cP_\Omega\cP_\cT\cP_\Omega)(\cP_\Omega-\cP_\Omega\cP_\cT\cP_\Omega)^{-1}\sign(\S^*)=\lambda\sign(\S^*)$. So to prove the dual condition \eqref{equ: dual conditions for W}, it suffices to show that
\begin{equation}
\label{equ: dual conditions for W^L}
\begin{split}
&\mbox{(a)}\quad \|\W^L\|\le \frac{\sigma_r(\X^*)}{4},\\
&\mbox{(b)}\quad \|\cP_\Omega(\X^*+\W^L)\|_F\le \frac{\lambda}{4},\\
&\mbox{(c)}\quad \|\cP_{\Omega^\perp}(\X^*+\W^L)\|_\infty\le \frac{\lambda}{4},
\end{split}
\end{equation}
\vspace{-0.4cm}
\begin{equation}
\label{equ: dual conditions for W^S}
\begin{split}
&\mbox{(d)}\quad \|\W^S\|\le \frac{\sigma_r(\X^*)}{4},\\
&\mbox{(e)}\quad \|\cP_{\Omega^\perp}\W^S\|_\infty\le \frac{\lambda}{4}.
\end{split}
\end{equation}

\subsection{Proof of Dual Conditions}
Since we have constructed the dual certificate $\W$, the remainder is to show that $\W$ obeys dual conditions \eqref{equ: dual conditions for W^L} and \eqref{equ: dual conditions for W^S} with high probability. We have the following.

\begin{lemma}
Assume $\Omega_j\sim\ber(q)$, $j=1,2,...,j_0$, and $j_0=2\lceil\log n_{(1)}\rceil$. Then under the other assumptions of Theorem \ref{theorem: robust PCA}, $\W^L$ given by \eqref{equ: golfing scheme for RPCA} obeys dual condition \eqref{equ: dual conditions for W^L}.
\end{lemma}

\begin{proof}
Let $\Z_j=\cP_\cT(\X^*-\Y_j)\in\cT$. Then we have
\begin{equation*}
\Z_j=\cP_\cT\Z_{j-1}-q^{-1}\cP_\cT\cP_{\Omega_j}\cP_\cT\Z_{j-1}=(\cP_{\cT}-q^{-1}\cP_\cT\cP_{\Omega_j}\cP_\cT)\Z_{j-1},
\end{equation*}
and $\Y_j=\sum_{k=1}^j q^{-1}\cP_{\Omega_k}\Z_{k-1}\in\Omega^\perp$. We set $q=\Omega(\epsilon^{-2}\mu r\log n_{(1)}/n_{(2)})$ with a small constant $\epsilon$.

\medskip
\noindent{\emph{Proof of (a).}} It holds that
\begin{equation*}
\begin{split}
\|\W^L\|&=\|\cP_{\cT^\perp}\Y_{j_0}\|\le \sum_{k=1}^{j_0}\|q^{-1}\cP_{\cT^\perp}\cP_{\Omega_k}\Z_{k-1}\|\\
&=\sum_{k=1}^{j_0}\|\cP_{\cT^\perp}(q^{-1}\cP_{\Omega_k}\Z_{k-1}-\Z_{k-1})\|\\
&\le \sum_{k=1}^{j_0}\|q^{-1}\cP_{\Omega_k}\Z_{k-1}-\Z_{k-1}\|\\
&\le C_0'\left(\frac{\log n_{(1)}}{q}\sum_{k=1}^{j_0}\|\Z_{k-1}\|_\infty+\sqrt{\frac{\log n_{(1)}}{q}}\sum_{k=1}^{j_0}\|\Z_{k-1}\|_{\infty,2}\right).\quad\mbox{(by Lemma \ref{lemma: 2 norm and infty norm})}
\end{split}
\end{equation*}
We note that by Lemma \ref{lemma: infty norm and infty norm},
\begin{equation*}
\|\Z_{k-1}\|_\infty\le \left(\frac{1}{2}\right)^{k-1}\|\Z_0\|_\infty,
\end{equation*}
and by Lemma \ref{lemma: infty 2 norm and infty, infty 2 norm},
\begin{equation*}
\|\Z_{k-1}\|_{\infty,2}\le \frac{1}{2}\sqrt{\frac{n_{(1)}}{\mu r}}\|\Z_{k-2}\|_\infty+\frac{1}{2}\|\Z_{k-2}\|_{\infty,2}.
\end{equation*}
Therefore,
\begin{equation*}
\begin{split}
\|\Z_{k-1}\|_{\infty,2}&\le \left(\frac{1}{2}\right)^{k-1} \sqrt{\frac{n_{(1)}}{\mu r}}\|\Z_0\|_\infty+\frac{1}{2}\|\Z_{k-2}\|_{\infty,2}\\
&\le (k-1)\left(\frac{1}{2}\right)^{k-1}\sqrt{\frac{n_{(1)}}{\mu r}}\|\Z_0\|_\infty+\left(\frac{1}{2}\right)^{k-1}\|\Z_0\|_{\infty,2},
\end{split}
\end{equation*}
and so we have
\begin{equation*}
\begin{split}
&\quad \|\W^L\|\\&\le C_0'\left[\frac{\log n_{(1)}}{q}\sum_{k=1}^{j_0}\left(\frac{1}{2}\right)^{k-1}\hspace{-0.2cm}\|\Z_0\|_\infty\hspace{-0.1cm}+\hspace{-0.1cm}\sqrt{\frac{\log n_{(1)}}{q}}\sum_{k=1}^{j_0}\left(\hspace{-0.1cm}(k-1)\hspace{-0.1cm}\left(\frac{1}{2}\right)^{k-1}\hspace{-0.2cm}\sqrt{\frac{n_{(1)}}{\mu r}}\|\Z_0\|_\infty\hspace{-0.1cm}+\hspace{-0.1cm}\left(\frac{1}{2}\right)^{k-1}\hspace{-0.2cm}\|\Z_0\|_{\infty,2}\right)\right]\\
&\le 2C_0'\left[\frac{\log n_{(1)}}{q}\|\X^*\|_\infty+\sqrt{\frac{n_{(1)}\log n_{(1)}}{q\mu r}}\|\X^*\|_\infty+\sqrt{\frac{\log n_{(1)}}{q}}\|\X^*\|_{\infty,2}\right]\\
&\le \frac{1}{16}\left[\frac{n_{(2)}}{\mu r}\|\X^*\|_\infty+\frac{\sqrt{n_{(1)}n_{(2)}}}{\mu r}\|\X^*\|_\infty+\sqrt{\frac{n_{(2)}}{\mu r}}\|\X^*\|_{\infty,2}\right]\quad\mbox{(since $q=\Omega(\mu r\log n_{(1)})/n_{(2)}$)}\\
&\le \frac{\sigma_r(\X^*)}{4},\quad\mbox{(by incoherence \eqref{equ: strong incoherence for RPCA})}
\end{split}
\end{equation*}
where we have used the fact that
\begin{equation*}
\|\X^*\|_{\infty,2}\le \sqrt{n_{(1)}}\|\X^*\|_\infty\le \sqrt{\frac{\mu r}{n_{(2)}}}\sigma_r(\X^*).
\end{equation*}

\medskip
\noindent{\emph{Proof of (b).}}
Because $\Y_{j_0}\in\Omega^\perp$, we have $\cP_\Omega(\X^*+\cP_{\cT^\perp}\Y_{j_0})=\cP_\Omega(\X^*-\cP_{\cT}\Y_{j_0})=\cP_\Omega\Z_{j_0}$. It then follows from Theorem \ref{theorem: concentration of operator} that for a properly chosen $t$,
\begin{equation*}
\begin{split}
\|\Z_{j_0}\|_F&\le t^{j_0}\|\X^*\|_F\\&\le t^{j_0}\sqrt{n_1n_2}\|\X^*\|_{\infty}\\&\le t^{j_0}\sqrt{n_1n_2}\sqrt{\frac{\mu r}{n_1n_2}}\sigma_r(\X^*)\\&\le \frac{\lambda}{8}.\quad (t^{j_0}\le e^{-2\log n_{(1)}}\le n_{(1)}^{-2})
\end{split}
\end{equation*}

\medskip
\noindent{\emph{Proof of (c).}}
By definition, we know that $\X^*+\W^L=\Z_{j_0}+\Y_{j_0}$. Since we have shown $\|\Z_{j_0}\|_F\le\lambda/8$, it suffices to prove $\|\Y_{j_0}\|_\infty\le \lambda/8$. We have
\begin{equation*}
\begin{split}
\|\Y_{j_0}\|_\infty&\le q^{-1}\sum_{k=1}^{j_0}\|\cP_{\Omega_k}\Z_{k-1}\|_\infty\\
&\le q^{-1}\sum_{k=1}^{j_0} \epsilon^{k-1}\|\X^*\|_\infty\quad\mbox{(by Lemma \ref{lemma: infty norm and infty norm})}\\
&\le \frac{n_{(2)}\epsilon^2}{C_0\mu r\log n_{(1)}}\sqrt{\frac{\mu r}{n_{(1)}n_{(2)}}}\sigma_r(\X^*)\quad \mbox{(by incoherence \eqref{equ: strong incoherence for RPCA})}\\
&\le \frac{\lambda}{8},
\end{split}
\end{equation*}
if we choose
$\epsilon= C\left(\frac{\mu r(\log n_{(1)})^2}{n_{(2)}}\right)^{1/4}$ for an absolute constant $C$. This can be true once the constant $\rho_r$ is sufficiently small.
\end{proof}

We now prove that $\W^S$ given by \eqref{equ: least squares for RPCA} obeys dual condition \eqref{equ: dual conditions for W^S}. We have the following.
\begin{lemma}
Assume $\Omega\sim\ber(p)$.
Then under the other assumptions of Theorem \ref{theorem: robust PCA}, $\W^S$ given by \eqref{equ: least squares for RPCA} obeys dual condition \eqref{equ: dual conditions for W^S}.
\end{lemma}

\begin{proof}
According to the standard de-randomization argument~\cite{Candes}, it is equivalent to studying the case when the signs $\delta_{ij}$ of $\S_{ij}^*$ are independently distributed as
\begin{equation*}
\delta_{ij}=
\begin{cases}
1, & \textup{\mbox{w.p. }} p/2,\\
0, & \textup{\mbox{w.p. }} 1-p,\\
-1, & \textup{\mbox{w.p. }} p/2.\\
\end{cases}
\end{equation*}

\noindent{\emph{Proof of (d).}}
Recall that
\begin{equation*}
\begin{split}
\W^S&=\lambda\cP_{\cT^\perp}\sum_{k\ge 0}(\cP_\Omega\cP_\cT\cP_\Omega)^k\sign(\S^*)\\
&=\lambda\cP_{\cT^\perp}\sign(\S^*)+\lambda\cP_{\cT^\perp}\sum_{k\ge 1}(\cP_\Omega\cP_\cT\cP_\Omega)^k\sign(\S^*).
\end{split}
\end{equation*}
To bound the first term, we have $\|\sign(\S^*)\|\le 4\sqrt{n_{(1)}p}$~\cite{vershynin2010introduction}. So $\|\lambda\cP_{\cT^\perp}\sign(\S^*)\|\le \lambda\|\sign(\S^*)\|\le 4\sqrt{p}\sigma_r(\X^*)\le \sigma_r(\X^*)/8$.

We now bound the second term. Let $\cG=\sum_{k\ge 1}(\cP_\Omega\cP_\cT\cP_\Omega)^k$, which is self-adjoint, and denote by $N_{n_1}$ and $N_{n_2}$ the $\frac{1}{2}$-nets of $\mathbb{S}^{n_1-1}$ and $\mathbb{S}^{n_1-1}$ of sizes at most $6^{n_1}$ and $6^{n_2}$, respectively~\cite{ledoux2005concentration}. We know that~[\cite{vershynin2010introduction}, Lemma 5.4]
\begin{equation*}
\begin{split}
\|\cG(\sign(\S^*))\|&=\sup_{\x\in\mathbb{S}^{n_2-1},\y\in\mathbb{S}^{n_1-1}}\langle\cG(\y\x^T),\sign(\S^*)\rangle\\&\le 4\sup_{\x\in N_{n_2},\y\in N_{n_1}}\langle\cG(\y\x^T),\sign(\S^*)\rangle.
\end{split}
\end{equation*}
Consider the random variable $X(\x,\y)=\langle\cG(\y\x^T),\sign(\S^*)\rangle$ which has zero expectation. By Hoeffding's inequality, we have
\begin{equation*}
\Pr(|X(\x,\y)|>t)\le 2\exp\left(-\frac{t^2}{2\|\cG(\x\y^T)\|_F^2}\right)\le 2\exp\left(-\frac{t^2}{2\|\cG\|^2}\right).
\end{equation*}
Therefore, by a union bound,
\begin{equation*}
\Pr(\|\cG(\sign(\S^*))\|>t)\le 2\times6^{n_1+n_2}\exp\left(-\frac{t^2}{8\|\cG\|^2}\right).
\end{equation*}
Note that conditioned on the event $\{\|\cP_\Omega\cP_\cT\|\le\sigma\}$, we have $\|\cG\|= \left\|\sum_{k\ge 1}(\cP_\Omega\cP_\cT\cP_\Omega)^k\right\|\le \frac{\sigma^2}{1-\sigma^2}$. So
\begin{equation*}
\Pr(\lambda\|\cG(\sign(\S^*))\|>t)\le 2\times6^{n_1+n_2}\exp\left(-\frac{t^2}{8\lambda^2}\left(\frac{1-\sigma^2}{\sigma^2}\right)^2\right)\Pr(\|\cP_\Omega\cP_\cT\|\le\sigma)+\Pr(\|\cP_\Omega\cP_\cT\|>\sigma).
\end{equation*}
Lemma \ref{lemma: small operator norm of P_Omega*P_T} guarantees that event $\{\|\cP_\Omega\cP_\cT\|\le\sigma\}$ holds with high probability for a very small absolute constant $\sigma$. Setting $t=\frac{\sigma_r(\X^*)}{8}$, this completes the proof of (d).
\end{proof}

\noindent{\emph{Proof of (e).}}
Recall that $\W^S=\lambda\cP_{\cT^\perp}\sum_{k\ge 0}(\cP_\Omega\cP_\cT\cP_\Omega)^k\sign(\S^*)$ and so
\begin{equation*}
\begin{split}
\cP_{\Omega^\perp}\W^S&=\lambda\cP_{\Omega^\perp}(\cI-\cP_\cT)\sum_{k\ge 0}(\cP_\Omega\cP_\cT\cP_\Omega)^k\sign(\S^*)\\
&=-\lambda\cP_{\Omega^\perp}\cP_\cT\sum_{k\ge 0}(\cP_\Omega\cP_\cT\cP_\Omega)^k\sign(\S^*).
\end{split}
\end{equation*}
Then for any $(i,j)\in\Omega^\perp$, we have
\begin{equation*}
\W_{ij}^S=\langle\W^S,\e_i\e_j^T\rangle=\left\langle\lambda\sign(\S^*),-\sum_{k\ge 0}(\cP_\Omega\cP_\cT\cP_\Omega)^k\cP_\Omega\cP_{\cT}(\e_i\e_j^T)\right\rangle.
\end{equation*}
Let $X(i,j)=-\sum_{k\ge 0}(\cP_\Omega\cP_\cT\cP_\Omega)^k\cP_\Omega\cP_{\cT}(\e_i\e_j^T)$. By Hoeffding's inequality and a union bound,
\begin{equation*}
\Pr\left(\sup_{ij}|\W_{ij}^S|>t\right)\le 2\sum_{ij}\exp\left(-\frac{2t^2}{\lambda^2\|X(i,j)\|_F^2}\right).
\end{equation*}
We note that conditioned on the event $\{\|\cP_\Omega\cP_\cT\|\le\sigma\}$, for any $(i,j)\in\Omega^\perp$,
\begin{equation*}
\begin{split}
\|X(i,j)\|_F&\le \frac{1}{1-\sigma^2}\sigma\|\cP_\cT(\e_i\e_j^T)\|_F\\
&\le \frac{1}{1-\sigma^2}\sigma\sqrt{1-\|\cP_{\cT^\perp}(\e_i\e_j^T)\|_F^2}\\
&= \frac{1}{1-\sigma^2}\sigma\sqrt{1-\|(\I-\U\U^T)\e_i\|_2^2\|(\I-\V\V^T)\e_j\|_2^2}\\
&\le \frac{1}{1-\sigma^2}\sigma\sqrt{1-\left(1-\frac{\mu r}{n_{(1)}}\right)\left(1-\frac{\mu r}{n_{(2)}}\right)}\\
&\le \frac{1}{1-\sigma^2}\sigma\sqrt{\frac{\mu r}{n_{(1)}}+\frac{\mu r}{n_{(2)}}}.
\end{split}
\end{equation*}
Then unconditionally,
\begin{equation*}
\Pr\left(\sup_{ij}|\W_{ij}^S|>t\right)\le 2n_{(1)}n_{(2)}\exp\left(-\frac{2t^2}{\lambda^2}\frac{(1-\sigma^2)^2n_{(1)}n_{(2)}}{\sigma^2\mu r(n_{(1)}+n_{(2)})}\right)\Pr(\|\cP_\Omega\cP_\cT\|\le\sigma)+\Pr(\|\cP_\Omega\cP_\cT\|>\sigma).
\end{equation*}
By Lemma \ref{lemma: small operator norm of P_Omega*P_T} and setting $t=\lambda/4$, the proof of (e) is completed.

\section{Proof of Theorem~\ref{theorem: lower bound for general optimization}}
\label{section: proof of lower bounds}

Our computational lower bound for problem (\textbf{P}) assumes the hardness of random 4-SAT.
\begin{conjecture}[Random 4-SAT] \label{conj:r4sat}
Let $c > \ln 2$ be a constant. Consider a random 4-SAT
formula on $n$ variables in which each clause has $4$ literals,
and in which each of the $16n^4$ clauses is picked independently
with probability $c/n^3$. Then any algorithm which always outputs
1 when the random formula is satisfiable, and outputs
0 with probability at least $1/2$ when the random formula is
unsatisfiable, must run in $2^{c' n}$ time on some input, where
$c' > 0$ is an absolute constant.
\end{conjecture}

Based on Conjecture \ref{conj:r4sat}, we have the following computational lower bound for problem (\textbf{P}). We show that problem (\textbf{P}) is in general hard for deterministic algorithms. If we additionally assume {\sf BPP = P}, then the same conclusion holds for randomized algorithms with high probability.

\medskip
\noindent{\textbf{Theorem~\ref{theorem: lower bound for general optimization}} (Computational Lower Bound. Restated)\textbf{.}}
\emph{Assume Conjecture~\ref{conj:r4sat}. Then there exists an absolute constant $\epsilon_0 > 0$ for which any algorithm that achieves $(1+\epsilon) \OPT$ in objective function value for problem \textup{(\textbf{P})} with $\epsilon \le \epsilon_0$, and with constant probability, requires $2^{\Omega(n_1 + n_2)}$ time, where $\OPT$ is the optimum. If in addition, {\sf BPP = P}, then the same conclusion holds for randomized algorithms succeeding with probability at least $2/3$.}

\begin{proof}
  Theorem~\ref{theorem: lower bound for general optimization} is proved by using the hypothesis that random 4-SAT is hard to show
  hardness of the Maximum Edge Biclique problem for deterministic algorithms.

\begin{definition}[Maximum Edge Biclique] The problem is
\begin{itemize}[leftmargin=40pt]
\item[Input:] An $n$-by-$n$ bipartite graph $G$.
\item[Output:] A $k_1$-by-$k_2$ complete bipartite subgraph of $G$, such that $k_1 \cdot k_2$ is maximized.
\end{itemize}
\end{definition}

\cite{goerdt2004approximation} showed that under the random 4-SAT assumption there exist two constants $\epsilon_1 > \epsilon_2 > 0$ such that no efficient deterministic algorithm is able to distinguish between bipartite graphs $G(U, V,E)$ with $|U| = |V | = n$ which have a clique of size $ \ge (n/16)^2 (1+\epsilon_1)$ and those in which all bipartite cliques are of size $\le (n/16)^2(1 + \epsilon_2)$. The reduction uses a bipartite graph $G$ with at least $t n^2$ edges with large probability, for a constant $t$.

Given a given bipartite graph $G(U, V,E)$, define $H(\cdot)$ as follows.
Define the matrix $\Y$ and $\W$: $\Y_{ij} = 1$ if edge $(U_i, V_j) \in E$, $\Y_{ij} = 0$ if edge $(U_i, V_j) \not\in E$; $\W_{ij} = 1$ if edge $(U_i, V_j) \in E$, and $\W_{ij} = \textnormal{poly}(n) $ if edge $(U_i, V_j) \not\in E$.
Choose a large enough constant $\beta>0$ and let $H(\A \B) = \beta \sum_{ij} \W_{ij}^2 (\Y_{ij} - (\A\B)_{ij})^2$.
Now, if there exists a biclique in $G$ with at least $(n/16)^2(1 + \epsilon_2)$ edges, then the number of remaining edges is at most $tn^2 - (n/16)^2(1 + \epsilon_1)$, and so the solution to $\min H(\A\B) + \frac{1}{2} \|\A\B\|_F^2$ has cost at most $ \beta[t n^2 - (n/16)^2(1 +\epsilon_1) ] + n^2$. On the other hand, if there does not exist a biclique that has more than $(n/16)^2(1 + \epsilon_2)$ edges, then the number of remaining edges is at least $(n/16)^2(1 + \epsilon_2)$, and so any solution to $\min H(\A\B) + \frac{1}{2} \|\A\B\|_F^2$ has cost at least $\beta[tn^2 - (n/16)^2 (1 + \epsilon_2)]$. Choose $\beta$ large enough so that $\beta[tn^2 - (n/16)^2 (1 + \epsilon_2)] > \beta[t n^2 - (n/16)^2(1 +\epsilon_1) ] + n^2$. This combined with the result in~\cite{goerdt2004approximation} completes the proof for deterministic algorithms.

To rule out randomized algorithms running in time $2^{\alpha(n_1 + n_2)}$ for some
function $\alpha$ of $n_1, n_2$ for which $\alpha = o(1)$, observe that we
can define a new problem which is the same as problem \textup{(\textbf{P})}
  except the input description of $H$ is padded with a string of $1$s
  of length $2^{(\alpha/2)(n_1 + n_2)}$. This string is irrelevant for solving
  problem \textup{(\textbf{P})} but changes the input size to
  $N = \poly(n_1, n_2) + 2^{(\alpha/2)(n_1+n_2)}$. By the argument in the previous
  paragraph, any deterministic algorithm still requires $2^{\Omega(n)} = N^{\omega(1)}$ time to solve this problem, which is super-polynomial in the new input
  size $N$. However, if a randomized algorithm can solve it in $2^{\alpha(n_1+n_2)}$ time, then it runs in $\poly(N)$ time. This contradicts the assumption
  that {\sf BPP = P}. This completes the proof.
\end{proof}

\section{Dual and Bi-Dual Problems}
\label{section: Dual and Bi-Dual Problems}
In this section, we derive the dual and bi-dual problems of non-convex program (\textbf{P}). According to \eqref{equ: Lagrangian form}, the primal problem (\textbf{P}) is equivalent to
\begin{equation*}
\min_{\A,\B} \max_\bLambda \frac{1}{2}\|-\bLambda-\A\B\|_F^2-\frac{1}{2}\|\bLambda\|_F^2-H^*(\bLambda).
\end{equation*}
Therefore, the dual problem is given by
\begin{equation*}
\begin{split}
&\quad \max_\bLambda \min_{\A,\B} \frac{1}{2}\|-\bLambda-\A\B\|_F^2-\frac{1}{2}\|\bLambda\|_F^2-H^*(\bLambda)\\
&=\max_\bLambda \frac{1}{2}\sum_{i=r+1}^{n_{(2)}}\sigma_i^2(-\bLambda)-\frac{1}{2}\|\bLambda\|_F^2-H^*(\bLambda)\\
&=\max_\bLambda -\frac{1}{2}\|\bLambda\|_r^2-H^*(\bLambda),\qquad\qquad (\textbf{D1})
\end{split}
\end{equation*}
where $\|\bLambda\|_r^2=\sum_{i=1}^r\sigma_i^2(\bLambda)$.
The bi-dual problem is derived by
\begin{equation*}
\begin{split}
&\quad\min_\M \max_{\bLambda,\bLambda'} -\frac{1}{2}\|\bLambda\|_r^2-H^*(\bLambda')+\langle\M,\bLambda'-\bLambda\rangle\\
&=\min_\M \max_{-\bLambda} \left[\langle\M,-\bLambda\rangle-\frac{1}{2}\|-\bLambda\|_r^2\right]+\max_{\bLambda'} \left[\langle\M,\bLambda'\rangle-H^*(\bLambda')\right]\\
&=\min_\M \|\M\|_{r*}+H(\M),\qquad\qquad (\textbf{D2})
\end{split}
\end{equation*}
where $\|\M\|_{r*}=\max_\X \langle\M,\X\rangle-\frac{1}{2}\|\X\|_r^2$ is a convex function, and $H(\M)=\max_{\bLambda'} \left[\langle\M,\bLambda'\rangle-H^*(\bLambda')\right]$ holds by the definition of conjugate function.

Problems (\textbf{D1}) and (\textbf{D2}) can be solved efficiently due to their convexity. In particular, \cite{grussler2016low} provided a computationally efficient algorithm to compute the proximal operators of functions $\frac{1}{2}\|\cdot\|_r^2$ and $\|\cdot\|_{r*}$. Hence, the Douglas-Rachford algorithm can find the global minimum up to an $\epsilon$ error in function value in time $\textsf{poly}(1/\epsilon)$~\cite{he20121}.

\section{Recovery under Bernoulli and Uniform Sampling Models}
\label{section: Equivalence of Bernoulli and Uniform Models}

We begin by arguing that a recovery result under the Bernoulli model with some probability automatically implies a
corresponding result for the uniform model with at least the same probability. The argument follows Section 7.1 of \cite{Candes}. For completeness, we provide the proof here.

Denote by $\Pr_{\mathsf{Unif}(m)}$ and $\Pr_{\mathsf{Ber}(p)}$ probabilities calculated under the uniform and Bernoulli models and let ``Success'' be the event that the algorithm succeeds. We have
\begin{equation*}
\begin{split}
\Pr\nolimits_{\mathsf{Ber}(p)}(\text{Success})&=\sum_{k=0}^{n_1n_2}\Pr\nolimits_{\mathsf{Ber}(p)}(\text{Success}\mid|\Omega|=k)\Pr\nolimits_{\mathsf{Ber}(p)}(|\Omega|=k)\\
&\le \sum_{k=0}^{m}\Pr\nolimits_{\mathsf{Unif}(k)}(\text{Success}\mid |\Omega|=k)\Pr\nolimits_{\mathsf{Ber}(p)}(|\Omega|=k)+\sum_{k=m+1}^{n_1n_2}\Pr\nolimits_{\mathsf{Ber}(p)}(|\Omega|=k)\\
&\le \Pr\nolimits_{\mathsf{Unif}(m)}(\text{Success})+\Pr\nolimits_{\mathsf{Ber}(p)}(|\Omega|>m),
\end{split}
\end{equation*}
where we have used the fact that for $k\le m$, $\Pr\nolimits_{\mathsf{Unif}(k)}(\text{Success})\le \Pr\nolimits_{\mathsf{Unif}(m)}(\text{Success})$, and that the conditional distribution of $|\Omega|$ is uniform. Thus
\begin{equation*}
\Pr\nolimits_{\mathsf{Unif}(m)}(\text{Success})\ge \Pr\nolimits_{\mathsf{Ber}(p)}(\text{Success})-\Pr\nolimits_{\mathsf{Ber}(p)}(|\Omega|>m).
\end{equation*}
Take $p=m/(n_1n_2)-\epsilon$, where $\epsilon>0$. The conclusion follows from $\Pr\nolimits_{\mathsf{Ber}(p)}(|\Omega|>m)\le e^{-\frac{\epsilon^2 n_1n_2}{2p}}$.

\comment{
\begin{lemma}
\label{lemma: number of sampling by Bernoulli}
Let $n$ be the number of Bernoulli trials and suppose that $\Omega\sim\textup{\textsf{Ber}}(m/n)$. Then with probability at least $1-n^{-10}$, $|\Omega|=\Theta(m)$, provided that $m\ge c\log n$ for an absolute constant $c$.
\end{lemma}
\begin{proof}
By the scalar Chernoff bound, with $\epsilon>0$ we have
\begin{equation}
\label{equ: scalar Chernoff bound <}
\Pr(|\Omega|\le m-n\epsilon)\le \exp\left(-\epsilon^2n^2/(2m)\right),
\end{equation}
and
\begin{equation}
\label{equ: scalar Chernoff bound >}
\Pr(|\Omega|\ge m+n\epsilon)\le \exp\left(-\epsilon^2n^2/(3m)\right).
\end{equation}
Taking $\epsilon=m/(2n)$ and $m\ge c_1\log n$ in \eqref{equ: scalar Chernoff bound <} for an appropriate absolute constant $c_1$, we have
\begin{equation}
\label{equ: |Omega| <}
\Pr(|\Omega|\le m/2)\le\exp(-m/4)\le \frac{n^{-10}}{2}.
\end{equation}
Taking $\epsilon=m/n$ and $m\ge c_2\log n$ in \eqref{equ: scalar Chernoff bound >} for an appropriate absolute constant $c_2$, we have
\begin{equation}
\label{equ: |Omega| >}
\Pr(|\Omega|\ge 2m)\le\exp(-m/3)\le \frac{n^{-10}}{2}.
\end{equation}
Given \eqref{equ: |Omega| <} and \eqref{equ: |Omega| >}, we conclude that $m/2<|\Omega|<2m$ with probability at least $1-n^{-10}$, provided that $m\ge c\log n$ for an absolute constant $c$.
\end{proof}
}

\bibliographystyle{alpha}
\bibliography{reference}

\newcommand{\etalchar}[1]{$^{#1}$}
\begin{thebibliography}{AAZB{\etalchar{+}}16}

\bibitem[AAZB{\etalchar{+}}16]{agarwal2016finding}
Naman Agarwal, Zeyuan Allen-Zhu, Brian Bullins, Elad Hazan, and Tengyu Ma.
\newblock Finding approximate local minima for nonconvex optimization in linear
  time.
\newblock {\em arXiv preprint arXiv:1611.01146}, 2016.

\bibitem[ABGM14]{arora2014provable}
Sanjeev Arora, Aditya Bhaskara, Rong Ge, and Tengyu Ma.
\newblock More algorithms for provable dictionary learning.
\newblock {\em arXiv preprint arXiv:1401.0579}, 2014.

\bibitem[ABHZ16]{awasthi2016learning}
Pranjal Awasthi, Maria-Florina Balcan, Nika Haghtalab, and Hongyang Zhang.
\newblock Learning and 1-bit compressed sensing under asymmetric noise.
\newblock In {\em Annual Conference on Learning Theory}, pages 152--192, 2016.

\bibitem[AG16]{anandkumar2016efficient}
Anima Anandkumar and Rong Ge.
\newblock Efficient approaches for escaping higher order saddle points in
  non-convex optimization.
\newblock {\em arXiv preprint arXiv:1602.05908}, 2016.

\bibitem[AGH{\etalchar{+}}14]{anandkumar2014tensor}
Animashree Anandkumar, Rong Ge, Daniel~J Hsu, Sham~M Kakade, and Matus
  Telgarsky.
\newblock Tensor decompositions for learning latent variable models.
\newblock {\em Journal of Machine Learning Research}, 15(1):2773--2832, 2014.

\bibitem[AMS11]{ambuhl2011inapproximability}
Christoph Amb{\"u}hl, Monaldo Mastrolilli, and Ola Svensson.
\newblock Inapproximability results for maximum edge biclique, minimum linear
  arrangement, and sparsest cut.
\newblock {\em SIAM Journal on Computing}, 40(2):567--596, 2011.

\bibitem[ANW12]{agarwal2012noisy}
Alekh Agarwal, Sahand Negahban, and Martin~J Wainwright.
\newblock Noisy matrix decomposition via convex relaxation: Optimal rates in
  high dimensions.
\newblock {\em The Annals of Statistics}, pages 1171--1197, 2012.

\bibitem[ARR14]{ahmed2014blind}
Ali Ahmed, Benjamin Recht, and Justin Romberg.
\newblock Blind deconvolution using convex programming.
\newblock {\em IEEE Transactions on Information Theory}, 60(3):1711--1732,
  2014.

\bibitem[AZ16]{allen2016katyusha}
Zeyuan Allen-Zhu.
\newblock Katyusha: The first direct acceleration of stochastic gradient
  methods.
\newblock {\em arXiv preprint arXiv:1603.05953}, 2016.

\bibitem[BCMV14]{bhaskara2014smoothed}
Aditya Bhaskara, Moses Charikar, Ankur Moitra, and Aravindan Vijayaraghavan.
\newblock Smoothed analysis of tensor decompositions.
\newblock In {\em ACM Symposium on Theory of Computing}, pages 594--603, 2014.

\bibitem[BE06]{beck2006strong}
Amir Beck and Yonina~C Eldar.
\newblock Strong duality in nonconvex quadratic optimization with two quadratic
  constraints.
\newblock {\em SIAM Journal on Optimization}, 17(3):844--860, 2006.

\bibitem[BH89]{Baldi1989neural}
Pierre Baldi and Kurt Hornik.
\newblock Neural networks and principal component analysis: Learning from
  examples without local minima.
\newblock {\em Neural Networks}, 2(1):53--58, 1989.

\bibitem[BKS16]{bhojanapalli2015dropping}
Srinadh Bhojanapalli, Anastasios Kyrillidis, and Sujay Sanghavi.
\newblock Dropping convexity for faster semi-definite optimization.
\newblock In {\em Annual Conference on Learning Theory}, pages 530--582, 2016.

\bibitem[BM05]{burer2005local}
Samuel Burer and Renato~DC Monteiro.
\newblock Local minima and convergence in low-rank semidefinite programming.
\newblock {\em Mathematical Programming}, 103(3):427--444, 2005.

\bibitem[BMP08]{bach2008convex}
Francis Bach, Julien Mairal, and Jean Ponce.
\newblock Convex sparse matrix factorizations.
\newblock {\em arXiv preprint arXiv:0812.1869}, 2008.

\bibitem[BNS16]{bhojanapalli2016global}
Srinadh Bhojanapalli, Behnam Neyshabur, and Nati Srebro.
\newblock Global optimality of local search for low rank matrix recovery.
\newblock In {\em Advances in Neural Information Processing Systems}, pages
  3873--3881, 2016.

\bibitem[BV04]{boyd2004convex}
Stephen Boyd and Lieven Vandenberghe.
\newblock {\em Convex optimization}.
\newblock Cambridge university press, 2004.

\bibitem[BZ16]{balcan2016noise}
Maria-Florina Balcan and Hongyang Zhang.
\newblock Noise-tolerant life-long matrix completion via adaptive sampling.
\newblock In {\em Advances in Neural Information Processing Systems}, pages
  2955--2963, 2016.

\bibitem[Che15]{chen2015incoherence}
Yudong Chen.
\newblock Incoherence-optimal matrix completion.
\newblock {\em IEEE Transactions on Information Theory}, 61(5):2909--2923,
  2015.

\bibitem[CHM{\etalchar{+}}15]{choromanska2015loss}
Anna Choromanska, Mikael Henaff, Michael Mathieu, G{\'e}rard~Ben Arous, and
  Yann LeCun.
\newblock The loss surfaces of multilayer networks.
\newblock In {\em International Conference on Artificial Intelligence and
  Statistics}, 2015.

\bibitem[CLMW11]{Candes}
Emmanuel~J. Cand\`{e}s, Xiaodong Li, Yi~Ma, and John Wright.
\newblock Robust principal component analysis?
\newblock {\em Journal of the ACM}, 58(3):11, 2011.

\bibitem[CR09]{Candes2009exact}
Emmanuel~J. Cand{\`e}s and Ben Recht.
\newblock Exact matrix completion via convex optimization.
\newblock {\em Foundations of Computational Mathematics}, 9(6):717--772, 2009.

\bibitem[CR13]{candes2013simple}
Emmanuel~J. Cand{\`e}s and Benjamin Recht.
\newblock Simple bounds for recovering low-complexity models.
\newblock {\em Mathematical Programming}, pages 1--13, 2013.

\bibitem[CRPW12]{chandrasekaran2012convex}
Venkat Chandrasekaran, Benjamin Recht, Pablo~A Parrilo, and Alan~S Willsky.
\newblock The convex geometry of linear inverse problems.
\newblock {\em Foundations of Computational Mathematics}, 12(6):805--849, 2012.

\bibitem[CT10]{Candes2010power}
Emmanuel~J. Cand{\`e}s and Terence Tao.
\newblock The power of convex relaxation: Near-optimal matrix completion.
\newblock {\em IEEE Transactions on Information Theory}, 56(5):2053--2080,
  2010.

\bibitem[CW15]{chen2015fast}
Yudong Chen and Martin~J. Wainwright.
\newblock Fast low-rank estimation by projected gradient descent: General
  statistical and algorithmic guarantees.
\newblock {\em arXiv preprint arXiv:1509.03025}, 2015.

\bibitem[Fei02]{f02}
Uriel Feige.
\newblock Relations between average case complexity and approximation
  complexity.
\newblock In {\em Proceedings of the 17th Annual {IEEE} Conference on
  Computational Complexity, Montr{\'{e}}al, Qu{\'{e}}bec, Canada, May 21-24,
  2002}, page~5, 2002.

\bibitem[GHJY15]{ge2015escaping}
Rong Ge, Furong Huang, Chi Jin, and Yang Yuan.
\newblock Escaping from saddle points -- online stochastic gradient for tensor
  decomposition.
\newblock In {\em Annual Conference on Learning Theory}, pages 797--842, 2015.

\bibitem[GJY17]{Rong2017No}
Rong Ge, Chi Jin, and Zheng Yi.
\newblock No spurious local minima in nonconvex low rank problems: A unified
  geometric analysis.
\newblock {\em arXiv preprint: 1704.00708}, 2017.

\bibitem[GL04]{goerdt2004approximation}
Andreas Goerdt and Andr{\'e} Lanka.
\newblock An approximation hardness result for bipartite clique.
\newblock In {\em Electronic Colloquium on Computational Complexity, Report},
  volume~48, 2004.

\bibitem[GLM16]{ge2016matrix}
Rong Ge, Jason~D Lee, and Tengyu Ma.
\newblock Matrix completion has no spurious local minimum.
\newblock In {\em Advances in Neural Information Processing Systems}, pages
  2973--2981, 2016.

\bibitem[GLZ17]{gamarnik2017matrix}
David Gamarnik, Quan Li, and Hongyi Zhang.
\newblock Matrix completion from {$O(n)$} samples in linear time.
\newblock {\em arXiv preprint arXiv:1702.02267}, 2017.

\bibitem[GRG16]{grussler2016low}
Christian Grussler, Anders Rantzer, and Pontus Giselsson.
\newblock Low-rank optimization with convex constraints.
\newblock {\em arXiv preprint arXiv:1606.01793}, 2016.

\bibitem[Gro11]{Gross2011recovering}
D.~Gross.
\newblock Recovering low-rank matrices from few coefficients in any basis.
\newblock {\em IEEE Transactions on Information Theory}, 57(3):1548--1566,
  2011.

\bibitem[GWL16]{gu2016low}
Quanquan Gu, Zhaoran Wang, and Han Liu.
\newblock Low-rank and sparse structure pursuit via alternating minimization.
\newblock In {\em International Conference on Artificial Intelligence and
  Statistics}, pages 600--609, 2016.

\bibitem[Har14]{hardt2014understanding}
Moritz Hardt.
\newblock Understanding alternating minimization for matrix completion.
\newblock In {\em IEEE Symposium on Foundations of Computer Science}, pages
  651--660, 2014.

\bibitem[HM12]{hardt2012algorithms}
Moritz Hardt and Ankur Moitra.
\newblock Algorithms and hardness for robust subspace recovery.
\newblock {\em arXiv preprint: 1211.1041}, 2012.

\bibitem[HMRW14]{hardt2014computational}
Moritz Hardt, Raghu Meka, Prasad Raghavendra, and Benjamin Weitz.
\newblock Computational limits for matrix completion.
\newblock In {\em Annual Conference on Learning Theory}, pages 703--725, 2014.

\bibitem[HV15]{haeffele2015global}
Benjamin~D Haeffele and Ren{\'e} Vidal.
\newblock Global optimality in tensor factorization, deep learning, and beyond.
\newblock {\em arXiv preprint arXiv:1506.07540}, 2015.

\bibitem[HY12]{he20121}
Bingsheng He and Xiaoming Yuan.
\newblock On the {$O(1/n)$} convergence rate of the douglas--rachford
  alternating direction method.
\newblock {\em SIAM Journal on Numerical Analysis}, 50(2):700--709, 2012.

\bibitem[HYV14]{haeffele2014structured}
Benjamin Haeffele, Eric Young, and Rene Vidal.
\newblock Structured low-rank matrix factorization: Optimality, algorithm, and
  applications to image processing.
\newblock In {\em International Conference on Machine Learning}, pages
  2007--2015, 2014.

\bibitem[IW97]{iw97}
Russell Impagliazzo and Avi Wigderson.
\newblock \emph{P = BPP} if \emph{E} requires exponential circuits:
  Derandomizing the {XOR} lemma.
\newblock In {\em {ACM} Symposium on the Theory of Computing}, pages 220--229,
  1997.

\bibitem[JBAS10]{journee2010low}
Michel Journ{\'e}e, Francis Bach, P-A Absil, and Rodolphe Sepulchre.
\newblock Low-rank optimization on the cone of positive semidefinite matrices.
\newblock {\em SIAM Journal on Optimization}, 20(5):2327--2351, 2010.

\bibitem[JGN{\etalchar{+}}17]{jin2017escape}
Chi Jin, Rong Ge, Praneeth Netrapalli, Sham~M Kakade, and Michael~I Jordan.
\newblock How to escape saddle points efficiently.
\newblock {\em arXiv preprint arXiv:1703.00887}, 2017.

\bibitem[JMD10]{jain2010guaranteed}
Prateek Jain, Raghu Meka, and Inderjit~S Dhillon.
\newblock Guaranteed rank minimization via singular value projection.
\newblock In {\em Advances in Neural Information Processing Systems}, pages
  937--945, 2010.

\bibitem[JNS13]{jain2013low}
Prateek Jain, Praneeth Netrapalli, and Sujay Sanghavi.
\newblock Low-rank matrix completion using alternating minimization.
\newblock In {\em ACM Symposium on Theory of Computing}, pages 665--674, 2013.

\bibitem[Kaw16]{kawaguchi2016deep}
Kenji Kawaguchi.
\newblock Deep learning without poor local minima.
\newblock {\em arXiv preprint arXiv:1605.07110}, 2016.

\bibitem[KBV09]{koren2009matrix}
Yehuda Koren, Robert Bell, and Chris Volinsky.
\newblock Matrix factorization techniques for recommender systems.
\newblock {\em IEEE Computer}, 42(8):30--37, 2009.

\bibitem[Kes12]{keshavan2012efficient}
Raghunandan~Hulikal Keshavan.
\newblock {\em Efficient algorithms for collaborative filtering}.
\newblock PhD thesis, Stanford University, 2012.

\bibitem[KMO10a]{keshavan2010matrix}
Raghunandan~H Keshavan, Andrea Montanari, and Sewoong Oh.
\newblock Matrix completion from a few entries.
\newblock {\em IEEE Transactions on Information Theory}, 56(6):2980--2998,
  2010.

\bibitem[KMO10b]{keshavan2010matrixnoise}
Raghunandan~H Keshavan, Andrea Montanari, and Sewoong Oh.
\newblock Matrix completion from noisy entries.
\newblock {\em Journal of Machine Learning Research}, 11:2057--2078, 2010.

\bibitem[Led05]{ledoux2005concentration}
Michel Ledoux.
\newblock {\em The concentration of measure phenomenon}.
\newblock Number~89. American Mathematical Society, 2005.

\bibitem[LLR16]{li2016recovery}
Yuanzhi Li, Yingyu Liang, and Andrej Risteski.
\newblock Recovery guarantee of weighted low-rank approximation via alternating
  minimization.
\newblock In {\em International Conference on Machine Learning}, pages
  2358--2367, 2016.

\bibitem[NNS{\etalchar{+}}14]{netrapalli2014non}
Praneeth Netrapalli, UN~Niranjan, Sujay Sanghavi, Animashree Anandkumar, and
  Prateek Jain.
\newblock Non-convex robust {PCA}.
\newblock In {\em Advances in Neural Information Processing Systems}, pages
  1107--1115, 2014.

\bibitem[NW12]{negahban2012restricted}
Sahand Negahban and Martin~J Wainwright.
\newblock Restricted strong convexity and weighted matrix completion: Optimal
  bounds with noise.
\newblock {\em Journal of Machine Learning Research}, 13:1665--1697, 2012.

\bibitem[OW92]{overton1992sum}
Michael~L Overton and Robert~S Womersley.
\newblock On the sum of the largest eigenvalues of a symmetric matrix.
\newblock {\em SIAM Journal on Matrix Analysis and Applications}, 13(1):41--45,
  1992.

\bibitem[Rec11]{recht2011simpler}
Benjamin Recht.
\newblock A simpler approach to matrix completion.
\newblock {\em Journal of Machine Learning Research}, 12:3413--3430, 2011.

\bibitem[RSW16]{razenshteyn2016weighted}
Ilya Razenshteyn, Zhao Song, and David~P. Woodruff.
\newblock Weighted low rank approximations with provable guarantees.
\newblock In {\em ACM Symposium on Theory of Computing}, pages 250--263, 2016.

\bibitem[SL15]{sun2015guaranteed}
Ruoyu Sun and Zhi-Quan Luo.
\newblock Guaranteed matrix completion via nonconvex factorization.
\newblock In {\em IEEE Symposium on Foundations of Computer Science}, pages
  270--289, 2015.

\bibitem[SQW16]{sun2016geometric}
Ju~Sun, Qing Qu, and John Wright.
\newblock A geometric analysis of phase retrieval.
\newblock In {\em IEEE International Symposium on Information Theory}, pages
  2379--2383, 2016.

\bibitem[SQW17a]{sun2016complete}
Ju~Sun, Qing Qu, and John Wright.
\newblock Complete dictionary recovery over the sphere {I}: Overview and the
  geometric picture.
\newblock {\em IEEE Transactions on Information Theory}, 63(2):853--884, 2017.

\bibitem[SQW17b]{sun2015completeII}
Ju~Sun, Qing Qu, and John Wright.
\newblock Complete dictionary recovery over the sphere {II}: Recovery by
  {Riemannian} trust-region method.
\newblock {\em IEEE Transactions on Information Theory}, 63(2):885--914, 2017.

\bibitem[SS05]{srebro2005rank}
Nathan Srebro and Adi Shraibman.
\newblock Rank, trace-norm and max-norm.
\newblock In {\em International Conference on Computational Learning Theory},
  pages 545--560. Springer, 2005.

\bibitem[SWZ14]{shen2014augmented}
Yuan Shen, Zaiwen Wen, and Yin Zhang.
\newblock Augmented lagrangian alternating direction method for matrix
  separation based on low-rank factorization.
\newblock {\em Optimization Methods and Software}, 29(2):239--263, 2014.

\bibitem[TBSR13]{tang2013compressed}
Gongguo Tang, Badri~Narayan Bhaskar, Parikshit Shah, and Benjamin Recht.
\newblock Compressed sensing off the grid.
\newblock {\em IEEE Transactions on Information Theory}, 59(11):7465--7490,
  2013.

\bibitem[TBSR15]{tu2015low}
Stephen Tu, Ross Boczar, Mahdi Soltanolkotabi, and Benjamin Recht.
\newblock Low-rank solutions of linear matrix equations via procrustes flow.
\newblock {\em arXiv preprint arXiv:1507.03566}, 2015.

\bibitem[Ver09]{Vershynin2010lectures}
Roman Vershynin.
\newblock Lectures in geometric functional analysis.
\newblock pages 1--76, 2009.

\bibitem[Ver10]{vershynin2010introduction}
Roman Vershynin.
\newblock Introduction to the non-asymptotic analysis of random matrices.
\newblock {\em arXiv preprint: 1011.3027}, 2010.

\bibitem[Ver15]{Vershynin2014Estimation}
Roman Vershynin.
\newblock Estimation in high dimensions: A geometric perspective.
\newblock In {\em Sampling theory, a renaissance}, pages 3--66. Springer, 2015.

\bibitem[WX12]{ICML2012Wang_233}
Yu-Xiang Wang and Huan Xu.
\newblock Stability of matrix factorization for collaborative filtering.
\newblock In {\em International Conference on Machine Learning}, pages
  417--424, 2012.

\bibitem[WYZ12]{wen2012solving}
Zaiwen Wen, Wotao Yin, and Yin Zhang.
\newblock Solving a low-rank factorization model for matrix completion by a
  nonlinear successive over-relaxation algorithm.
\newblock {\em Mathematical Programming Computation}, 4(4):333--361, 2012.

\bibitem[YPCC16]{yi2016fast}
Xinyang Yi, Dohyung Park, Yudong Chen, and Constantine Caramanis.
\newblock Fast algorithms for robust {PCA} via gradient descent.
\newblock In {\em Advances in neural information processing systems}, pages
  4152--4160, 2016.

\bibitem[ZL15]{zheng2015convergent}
Qinqing Zheng and John Lafferty.
\newblock A convergent gradient descent algorithm for rank minimization and
  semidefinite programming from random linear measurements.
\newblock In {\em Advances in Neural Information Processing Systems}, pages
  109--117, 2015.

\bibitem[ZL16]{zheng2016convergence}
Qinqing Zheng and John Lafferty.
\newblock Convergence analysis for rectangular matrix completion using
  {Burer-Monteiro} factorization and gradient descent.
\newblock {\em arXiv preprint arXiv:1605.07051}, 2016.

\bibitem[ZLZ13]{Zhang:Counterexample}
Hongyang Zhang, Zhouchen Lin, and Chao Zhang.
\newblock A counterexample for the validity of using nuclear norm as a convex
  surrogate of rank.
\newblock In {\em European Conference on Machine Learning and Principles and
  Practice of Knowledge Discovery in Databases}, volume 8189, pages 226--241,
  2013.

\bibitem[ZLZ16]{zhang2016completing}
Hongyang Zhang, Zhouchen Lin, and Chao Zhang.
\newblock Completing low-rank matrices with corrupted samples from few
  coefficients in general basis.
\newblock {\em IEEE Transactions on Information Theory}, 62(8):4748--4768,
  2016.

\bibitem[ZLZC15]{Zhang2015AAAI}
Hongyang Zhang, Zhouchen Lin, Chao Zhang, and Edward Chang.
\newblock Exact recoverability of robust {PCA} via outlier pursuit with tight
  recovery bounds.
\newblock In {\em AAAI Conference on Artificial Intelligence}, pages
  3143--3149, 2015.

\bibitem[ZLZG14]{Zhang:RobustLatLRR}
Hongyang Zhang, Zhouchen Lin, Chao Zhang, and Junbin Gao.
\newblock Robust latent low rank representation for subspace clustering.
\newblock {\em Neurocomputing}, 145:369--373, 2014.

\bibitem[ZLZG15]{zhang2015relations}
Hongyang Zhang, Zhouchen Lin, Chao Zhang, and Junbin Gao.
\newblock Relations among some low rank subspace recovery models.
\newblock {\em Neural Computation}, 27:1915--1950, 2015.

\bibitem[ZWG17]{zhang2017nonconvex}
Xiao Zhang, Lingxiao Wang, and Quanquan Gu.
\newblock A nonconvex free lunch for low-rank plus sparse matrix recovery.
\newblock {\em arXiv preprint arXiv:1702.06525}, 2017.

\bibitem[ZWL15]{zhao2015nonconvex}
Tuo Zhao, Zhaoran Wang, and Han Liu.
\newblock A nonconvex optimization framework for low rank matrix estimation.
\newblock In {\em Advances in Neural Information Processing Systems}, pages
  559--567, 2015.

\end{thebibliography}

\end{document}